\documentclass[journal, 10pt]{IEEEtran}
\usepackage{amsfonts}%
\usepackage{amssymb}%
\usepackage{setspace}
\usepackage[dvips]{graphicx}%
\usepackage{amsmath}%
\usepackage[usenames]{color}
\usepackage{algorithmic}
\usepackage{cite}
\usepackage{array}
\newtheorem{theorem}{Theorem}
\newtheorem{corollary}{Corollary}
\newtheorem{lemma}{Lemma}
\newtheorem{assumption}{Assumption}

\newtheorem{example-non*}{Example}
\newtheorem{example}{Example}

\newtheorem{remark}{Remark}
\allowdisplaybreaks %breaks long equations if it does not fit into one page.

\DeclareMathOperator*{\argmax}{arg\,max}

 \newcommand{\comment}[1]{}

\ifodd 1
\newcommand{\remove}[1]{}
\newcommand{\add}[1]{#1}
\else
\newcommand{\remove}[1]{#1}
\newcommand{\add}[1]{}
\fi

\ifodd 0
\newcommand{\bremove}[1]{}
\newcommand{\badd}[1]{#1}
\else
\newcommand{\bremove}[1]{#1}
\newcommand{\badd}[1]{}
\fi

\ifodd 1
\newcommand{\rmv}[1]{}
\else
\newcommand{\rmv}[1]{{\color{red}#1}}
\fi

\ifodd 0
\newcommand{\newc}[1]{{\color{blue}#1}} %revise of the text
\else
\newcommand{\newc}[1]{#1}
\fi

\ifodd 1
\newcommand{\janrmv}[1]{} %revise of the text
\else
\newcommand{\janrmv}[1]{#1}
\fi
\ifodd 0
\newcommand{\rev}[1]{{\color{blue}#1}} %revise of the text
\newcommand{\com}[1]{\textbf{\color{red}(COMMENT: #1)}} %comment of the text
\newcommand{\clar}[1]{\textbf{\color{green}(NEED CLARIFICATION: #1)}}

\else
\newcommand{\rev}[1]{#1}

\newcommand{\com}[1]{}
\newcommand{\clar}[1]{}
\fi

\ifodd 1
\else

\fi

\ifodd 0
\newcommand{\jremove}[1]{}
\newcommand{\aremove}[1]{{\color{red}#1}} %things that are in journal but not in online appendix.
\else
\newcommand{\jremove}[1]{{\color{black}#1}} %things that are in online appendix but not in paper.
\newcommand{\aremove}[1]{}
\fi
\begin{document}
\title{Contextual Online Learning for Multimedia Content Aggregation}
\author{\IEEEauthorblockN{Cem Tekin,~\IEEEmembership{Member,~IEEE}, Mihaela van der Schaar,~\IEEEmembership{Fellow,~IEEE}\\}
\thanks{Copyright (c) 2015 IEEE. Personal use of this material is permitted. However, permission to use this material for any other purposes must be obtained from the IEEE by sending a request to pubs-permissions@ieee.org.}
\thanks{This work is partially supported by the grants NSF CNS 1016081 and
AFOSR DDDAS.}
\thanks{C. Tekin and Mihaela van der Schaar are in Department of Electrical Engineering, UCLA, Los Angeles, CA, 90095. Email: cmtkn@ucla.edu, mihaela@ee.ucla.edu.}
\thanks{This online technical report is an extended version of the paper that appeared in IEEE Transactions on Multimedia \cite{tekinTMM2015}.}
}
%\author{\IEEEauthorblockN{Cem Tekin, Mingyan Liu\\}
%\IEEEauthorblockA{Department of Electrical Engineering and Computer Science\\
%University of Michigan, Ann Arbor, Michigan, 48109-2122\\
%Email: \{cmtkn, mingyan\}@umich.edu}
%}

\maketitle

\begin{abstract}
The last decade has witnessed a tremendous growth
in the volume as well as the diversity of multimedia content generated by a multitude of sources (news agencies, social media, etc.). Faced with a variety of content choices, consumers are exhibiting diverse preferences for content; their preferences often depend on the context in which they consume content as well as various exogenous events. To satisfy the consumers' demand for such diverse content, multimedia content aggregators (CAs) have emerged which gather content from numerous multimedia sources. 
A key challenge for such systems is to accurately predict what type of content each of its consumers prefers in a certain context, and adapt these predictions to the evolving consumers' preferences, contexts and content characteristics. We propose a novel, distributed, online
multimedia content aggregation framework, which gathers content generated by multiple heterogeneous producers to fulfill its consumers' demand for content.  Since both the multimedia content
characteristics and the consumers' preferences and contexts are unknown, the optimal content
aggregation strategy is unknown a priori. Our proposed content
aggregation algorithm is able to learn online what content to gather and how to match content and users by exploiting similarities
between consumer types. We prove bounds for our proposed learning algorithms that guarantee both the accuracy of the predictions as well as the learning speed. Importantly, our algorithms operate efficiently even when feedback from consumers is missing or content and preferences evolve over time. Illustrative results highlight the merits of the proposed content aggregation system in a variety of settings. 
\end{abstract}


\vspace{-0.1in}
\section{Introduction}\label{sec:intro}

A plethora of multimedia applications (web-based TV \cite{ren2012pricing, song2012advanced}, personalized video retrieval \cite{xu2008novel}, personalized news aggregation \cite{li2010contextual}, etc.) are emerging which require matching multimedia content generated by distributed sources with consumers exhibiting different interests. The matching is often performed by CAs (e.g., Dailymotion, Metacafe \cite{saxena2008analyzing}) that are responsible for mining the content of numerous multimedia sources in search of finding content which is interesting for the users. Both the characteristics of the content and preference of the consumers are evolving over time. An example of the system with users, CAs and multimedia sources is given in Fig. \ref{fig:systemmodel}.

\begin{figure}
\begin{center}
\includegraphics[width=0.95\columnwidth]{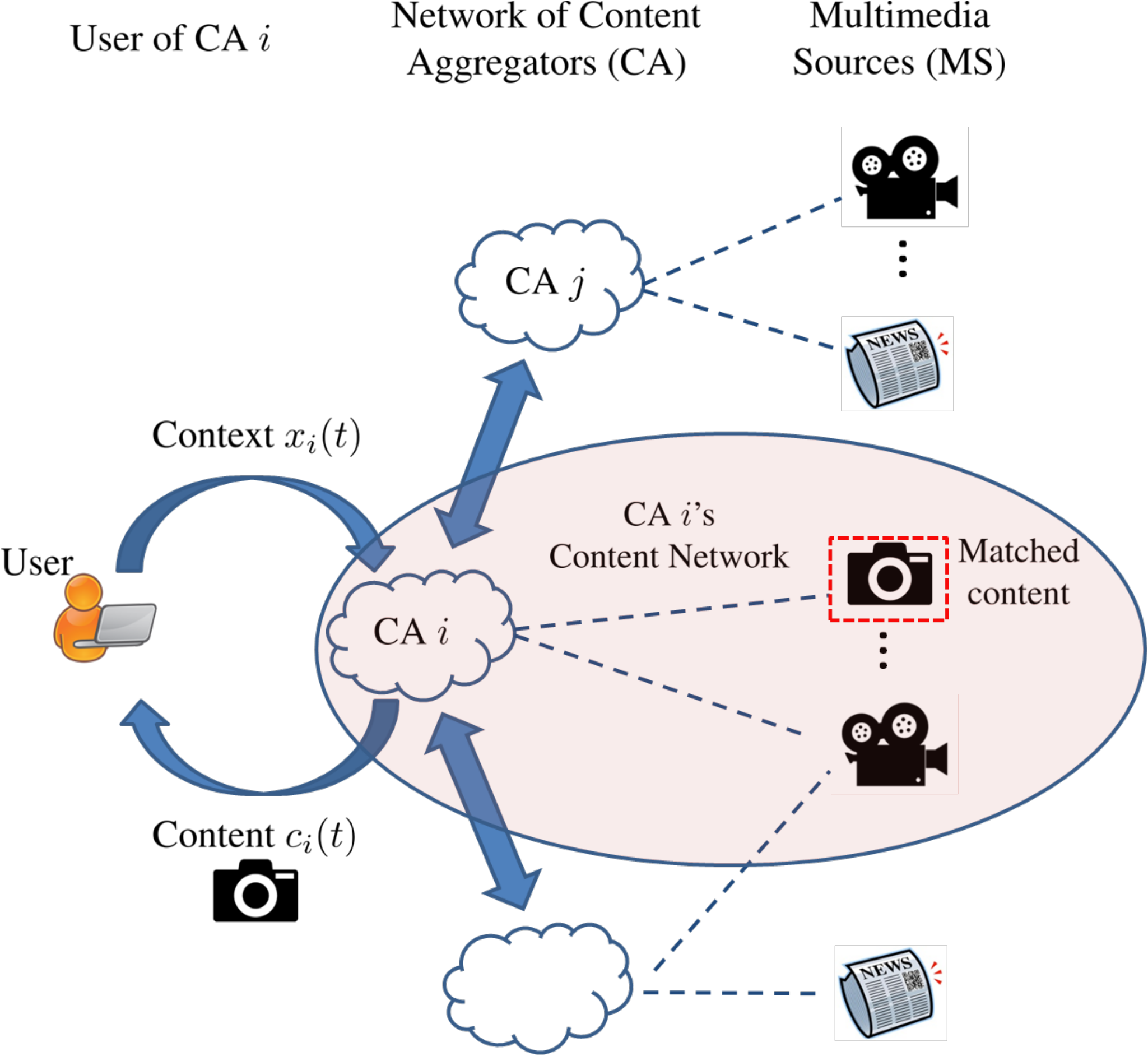}
\vspace{-0.1in}
\caption{Operation of the distributed content aggregation system. (i) A user with type/context $x_i(t)$ arrives to Content Aggregator (CA) $i$, (ii) CA $i$ chooses a {\em matching action} (either requests content from another CA or requests content from a multimedia source in its own network).}
%\vspace{-0.4in}
\label{fig:systemmodel}
\end{center}
\vspace{-0.2in}
\end{figure}

Each user is characterized by its context, which is a real-valued vector, that provides information about the users' content preferences. 
We assume a model where users arrive sequentially to a CA, and based on the type (context) of the user, the CA requests content from either one of the multimedia sources that it is connected to or from another CA that it is connected to. 
The context can represent information such as age, gender, search query, previously consumed content, etc. It may also represent the type of the device that the user is using \cite{mohan1999adapting} (e.g., PDA, PC, mobile phone).
The CA's role is to match its user with the most suitable content, which can be accomplished by requesting content from the most suitable multimedia source.\footnote{Although we use the term {\em request} to explain how content from a multimedia source is mined, our proposed method works also when a CA extracts the content from the multimedia source, without any decision making performed by the multimedia source.}
Since both the content generated by the multimedia sources and the user's characteristics change over time, it is unknown to the CA which multimedia source to match with the user. 
This problem can be formulated as an online learning problem, where the CA learns the best matching by exploring matchings of users with different content providers.
After a particular content matching is made, the user ``consumes" the content, and provides feedback/rating, such as {\em like} or {\em dislike}.\footnote{Our framework also works when the feedback is missing for some users.}
It is this feedback that helps a CA learn the preferences of its users and the characteristics of the content that is provided by the multimedia sources. 
Since this is a learning problem we equivalently call a CA, a content learner or simply,  a {\em learner}. 

\newc{Two possible real-world applications of content aggregation are business news aggregation and music aggregation. Business news aggregators can collect information from a variety of multinational and multilingual sources and make personalized recommendations to  specific individuals/companies based on their unique needs (see e.g. \cite{schranz2005building}). Music aggregators enable matching listeners with music content they enjoy both within the content network of the listeners as well as outside this network. For instance, distributed music aggregators can facilitate the sharing of music collections owned by diverse users without the need for centralized content manager/moderator/providers (see e.g. \cite{FStream}). A discussion of how these applications can be modeled using our framework is given in Section \ref{sec:probform}. Moreover, our proposed methods are tested on real-world datasets related to news aggregation and music aggregation in Section \ref{sec:numerical}. }

For each CA $i$, there are two types of users: {\em direct} and {\em indirect}. Direct users are the users that visit the website of CA $i$ to search for content. {\em Indirect} users are the users of another CA that requests content from CA $i$. 
A CA's goal is to maximize the number of likes received from its users (both direct and indirect).
This objective can be achieved by all CAs by the following {\em distributed} learning method: all CAs learn online which {\em matching action} to take for its current user, i.e., obtain content from a multimedia source that is directly connected, or request content from another CA.  
However, it is not trivial how to use the past information collected by the CAs in an efficient way, due to the vast number of contexts (different user types) and dynamically changing user and content characteristics. 
For instance, a certain type of content may become popular among users at a certain point in time, which will require the CA to obtain content from the multimedia source that generates that type of content. 

To jointly optimize the performance of the multimedia content aggregation system, we propose an online learning methodology that builds on contextual bandits \cite{slivkins2009contextual,lu2010contextual}. 
The performance of the proposed methodology is evaluated using the notion of regret: the difference between the expected total reward (number of content likes minus costs of obtaining the content) of the best content matching strategy given complete knowledge about the user preferences and content characteristics and the expected total reward of the algorithm used by the CAs.
When the user preferences and content characteristics are static, our proposed algorithms achieve sublinear regret in the number of users that have arrived to the system.\footnote{We use index $t$ to denote the number of users that have arrived so far. We also call $t$ the time index, and assume that one user arrives at each time step.} When the user preferences and content characteristics are slowly changing over time, our proposed algorithms achieve  $\epsilon$ time-averaged regret, where $\epsilon>0$ depends on the rate of change of the user and content characteristics.

The remainder of the paper is organized as follows. In Section \ref{sec:related}, we describe the related work and highlight the differences from our work. In Section \ref{sec:probform}, we describe the decentralized content aggregation problem, the optimal content matching scheme given the complete system model, and the regret of a learning algorithm with respect to the optimal content matching scheme. 
Then, we consider the model with unknown, static user preferences and content characteristics and propose a distributed online learning algorithm in Section \ref{sec:iid}. 
The analysis of the unknown, dynamic user preferences and content characteristics are given in Section \ref{sec:dynamic}.
Using real-world datasets, we provide numerical results on the performance of our distributed online learning algorithms in Section \ref{sec:numerical}. Finally, the concluding remarks are given in Section \ref{sec:conc}.

\vspace{-0.1in}
\section{Related Work} \label{sec:related}
Related work can be categorized into two: related work on recommender systems and related work on online learning methods called {\em multi-armed bandits}. 

\vspace{-0.1in}
\subsection{Related work on recommender systems and content matching}

A recommender system recommends items to its users based on the characteristics of the users and the items. The goal of a recommender system is to learn which users like which items, and recommend items such that the number of likes is maximized. 
For instance, in \cite{li2010contextual, deshpande2012linear} a recommender system that learns the preferences of its users in an online way based on the ratings submitted by the users is provided. It is assumed that the true relevance score of an item for a user is a linear function of the context of the user and the features of the item. Under this assumption, an online learning algorithm is proposed. 
In contrast, we consider a different model, where the relevance score need not be linear in the context. Moreover, due to the distributed nature of the problem we consider, our online learning algorithms need an additional phase called the training phase, which accounts for the fact that the CAs are uncertain about the information of the other aggregators that they are linked with. We focus on the long run performance and show that the regret per unit time approaches zero when the user and content characteristics are static.
An online learning algorithm for a centralized recommender which updates its recommendations as both the preferences of the users and the characteristics of items change over time is proposed in \cite{kohli2013fast}.

The general framework which exploits the similarities between the past users and the current user to recommend content to the current user is called collaborative filtering \cite{ sahoo2012hidden, linden2003amazon, miyahara2000collaborative}. These methods find the similarities between the current user and the past users by examining their search and feedback patterns, and then based on the interactions with the past {\em similar} users, matches the user with the content that has the highest estimated relevance score. For example, the most relevant content can be the content that is liked the highest number of times by similar users. 
Groups of similar users can be created by various methods such as clustering \cite{linden2003amazon}, and then, the matching will be made based on the content matched with the past users that are in the same group. 
%Another idea is to calculate relevance score after context-based filtering, and match a user with an item that have the highest relevance score for the user's context \cite{panniello2012comparing}. 
%Our proposed learning algorithms use a similar idea: partition the context space of users into hypercubes, and then learn the optimal matching over each hypercube. In this way, our work can also be considered as a type of collaborative filtering algorithm because each user is split into a specific partition of a hypercube, and is then treated the same as other users within that partition. However, the way that the aggregators learn the best content to match with each partition builds on the multia-armed bandit literature, through the separate training, exploration, and exploitation phases. Thus although our algorithm has a simple method of separating users into different groups, it features a groundbreaking method of matching content with each group.

The most striking difference between our content matching system and previously proposed is that in prior works, there is a central CA which knows the entire set of different types of content, and all the users arrive to this central CA. In contrast, we consider a decentralized system consisting of many CAs, many multimedia sources that these CAs are connected to, and heterogeneous user arrivals to these CAs. 
These CAs are cooperating with each other by only knowing the connections with their own neighbors but not the entire network topology. Hence, a CA does not know which multimedia sources another CA is connected to, but it learns over time whether that CA has access to content that the users like or not.  
Thus, our model can be viewed as a giant collection of individual CAs that are running in parallel. 
%An interaction between these aggregators happen when one of them request content from another one for its own user. 
%These interactions help an aggregator to match its users with content that is more relevant, from a much wider range of multimedia sources that the aggregator even does not need to be aware of. 

Another line of work \cite{roy2012empowering, roy2013towards} uses social streams mined in one domain, e.g., Twitter, to build a topic space that relates these streams to content in the multimedia domain. For example, in \cite{roy2012empowering}, Tweet streams are used to provide video recommendations in a commercial video search engine. A content adaptation method is proposed in \cite{mohan1999adapting} which enables the users with different types of contexts and devices to receive content that is in a suitable format to be accessed.
Video popularity prediction is studied in \cite{roy2013towards}, where the goal is to predict if a video will become popular in the multimedia domain, by detecting social trends in another social media domain (such as Twitter), and transferring this knowledge to the multimedia domain. 
Although these methods are very different from our methods, the idea of transferring knowledge from one multimedia domain to another can be carried out by CAs specialized in specific types of cross-domain content matching 
%(see Fig. \ref{fig:crossdomain}). In this setting, the CA will combine two sets of contexts, i.e., the context of the user and the context of the domain to decide on the content matching. 
For instance, one CA may transfer knowledge from tweets to predict the content which will have a high relevance/popularity for a user with a particular context, while another CA may scan through the Facebook posts of the user's friends to calculate the context of the domain in addition to the context of the user, and provide a matching according to this.

\newc{The advantages of our proposed approach over prior work in recommender systems are: (i) systematic analysis of recommendations' performance, including confidence bounds on the accuracy of the recommendations; (ii) no need for a priori knowledge of the users' preferences (i.e., system learns on-the-fly); (iii) achieve high accuracy even when the users' characteristics and content characteristics are changing over time; (iv) all these features are enabled in a network of distributed CAs.}

\newc{The differences of our work from the prior work in recommender systems is summarized in Table \ref{tab:comparison}.}

\newc{
\begin{table}
\centering
{
{\fontsize{9}{9}\selectfont
\setlength{\tabcolsep}{.1em}
\begin{tabular}{|c|c|c|c|c|c|}
\hline
 & Our work & \cite{li2010contextual, deshpande2012linear} & \cite{linden2003amazon} & \cite{sahoo2012hidden} & \cite{bouneffouf2012hybrid}\\
\hline
Distributed & Yes & No & No & No & No\\
%\hline
Reward model & H\"{o}lder & Linear & N/A & N/A & N/A \\
%\hline
Confidence bounds & Yes & No & No & No & No \\
%\hline
Regret bound & Yes  & Yes  & No & No & No \\
%\hline
Dynamic user & Yes & No & Yes & Yes & Yes\\
/content distribution & & & & & \\
\hline
\end{tabular}
}
}
\caption{Comparison of our work with other work in recommender systems}
\vspace{-0.3in}
\label{tab:comparison}
\end{table}
}
 
\comment{
\begin{figure}
\begin{center}
\includegraphics[width=0.95\columnwidth]{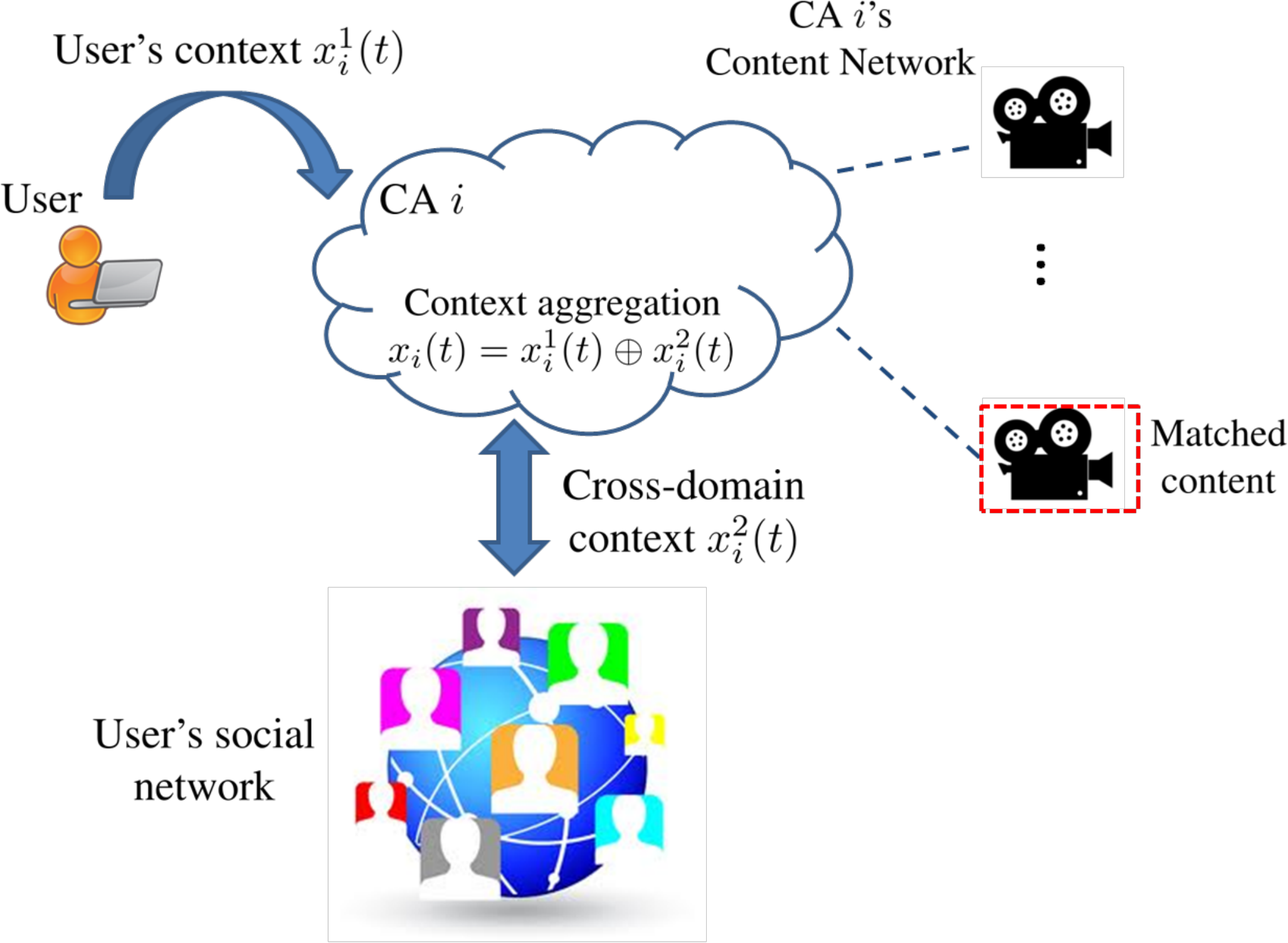}
\vspace{-0.1in}
\caption{Aggregation of user's context with the domain context to match the user with a content in another multimedia domain.} 
%\vspace{-0.4in}
\label{fig:crossdomain}
\end{center}
\vspace{-0.2in}
\end{figure}
}

\vspace{-0.1in}
\subsection{Related Work on Multi-armed Bandits}

Other than distributed content recommendation, our learning framework can be applied to any problem that can be formulated as a decentralized contextual bandit problem. Contextual bandits have been studied before in \cite{hazan2007online, slivkins2009contextual, dudik2011efficient, chu2011contextual, lu2010contextual} in a single agent setting, where the agent sequentially chooses from a set of alternatives with unknown rewards, and the rewards depend on the context information provided to the agent at each time step. In \cite{li2010contextual}, a contextual bandit algorithm named LinUCB is proposed for recommending personalized news articles,
which is variant of the UCB algorithm \cite{auer} designed for linear payoffs.
Numerical results on real-world Internet data are provided, but no
theoretical results on the resulting regret are derived.
The main difference of our work from single agent contextual bandits is that: (i) a three phase learning algorithm with {\em training}, {\em exploration} and {\em exploitation} phases is needed instead of the standard two phase, i.e., {\em exploration} and {\em exploitation} phases, algorithms used in centralized contextual bandit problems; (ii) the adaptive partitions of the context space should be formed in a way that each learner/aggregator can efficiently utilize what is learned by other learners about the same context; (iii) the algorithm is robust to missing feedback (some users do not rate the content).

\vspace{-0.1in}
\section{Problem Formulation}\label{sec:probform}

The system model is shown in Fig. \ref{fig:systemmodel}. There are $M$ content aggregators (CAs) which are indexed by the set ${\cal M} := \{1,2,\ldots,M\}$.
We also call each CA a {\em learner} since it needs to learn which type of content to provide to its users. 
Let ${\cal M}_{-i} := {\cal M} - \{i\}$ be the set of CAs that CA $i$ can choose from to request content.
Each CA has access to the contents over its content network as shown in Fig. \ref{fig:systemmodel}. The set of contents in CA $i$'s content network is denoted by ${\cal C}_i$. 
The set of all contents is denoted by ${\cal C} := \cup_{i \in {\cal M}} {\cal C}_i$.
The system works in a discrete time setting $t=1,2,\ldots,T$, where the following events happen sequentially, in each time slot: (i) a user with context $x_i(t)$ arrives to each CA $i \in {\cal M}$,\footnote{Although in this model user arrivals are synchronous, our framework will work for asynchronous user arrivals as well.} (ii) based on the context of its user each CA matches its user with a content (either from its own content network or by requesting content from another CA), (iii) the user provides a feedback, denoted by $y_i(t)$, which is either {\em like} ($y_i(t)=1$) or {\em dislike} ($y_i(t)=0$).

The set of {\em content matching actions} of CA $i$ is denoted by ${\cal K}_i := {\cal C}_i \cup {\cal M}_{-i}$.
Let ${\cal X} = [0,1]^d$ be the context space,\footnote{In general, our results will hold for any bounded subspace of $\mathbb{R}^n$.}  where $d$ is the dimension of the context space.
The context can include many properties of the user such as age, gender, income, previously liked content, etc. We assume that all these quantities are mapped into $[0,1]^d$. For instance, this mapping can be established by feature extraction methods such as the one given in \cite{li2010contextual}.
Another method is to represent each property of a user by a real number between $[0,1]$ (e.g., normalize the age by a maximum possible age, represent gender by set $\{0,1\}$, etc.), without feature extraction.
The feedback set of a user is denoted by ${\cal Y} := \{ 0, 1 \}$. Let $C_{\max} := \max_{i \in {\cal M}} |{\cal C}_i|$. We assume that all CAs know $C_{\max}$ but they do not need to know the content networks of other CAs. 

\newc{The following two examples demonstrate how business news aggregation and music aggregation fits our problem formulation.}

\begin{example}\label{example:1}
\newc{\textbf{Business news aggregation}. Consider a distributed set of news aggregators that operate in different countries (for instance a European news aggregator network as in \cite{schranz2005building}). Each news aggregator's content network (as portrayed in Fig. 1 of the manuscript) consists of content producers (multimedia sources) that are located in specific regions/countries. Consider a user with context $x$ (e.g. age, gender, nationality, profession) who subscribes to the CA $A$, which is located in the country where the user lives. This CA has access to content from local producers in that country but it can also request content from other CAs,  located in different countries. Hence, a CA has access to (local) content generated in other countries. In such scenarios, our proposed system is able to recommend to the user subscribing to CA $A$ also content from other CAs, by discovering the content that is most relevant to that user (based on its context $x$) across the entire network of CAs. For instance, for a user doing business in the transportation industry, our content aggregator system may learn to recommend road construction news, accidents or gas prices from particular regions that are on the route of the transportation network of the user.}
\end{example}

\begin{example}\label{example:2}
\newc{\textbf{Music aggregation}. Consider a distributed set of music aggregators that are specialized in specific genres of music: classical, jazz, rock, rap, etc. Our proposed model allows music aggregators to share content to provide personalized recommendation for a specific user. For instance, a user that subscribes (frequents/listens) to the classical music aggregator may also like specific jazz tracks. Our proposed system is able to discover and recommend to that user also other music that it will enjoy in addition to the music available to/owned by in aggregator to which it subscribes.}
\end{example}

\vspace{-0.1in}
\subsection{User and Content Characteristics}\label{sec:fixedprediction}

In this paper we consider two types of user and content characteristics. First, we consider the case when the user and content characteristics are {\em static}, i.e., they do not change over time. 
For this case, for a user with context $x$, $\pi_{c}(x)$ denotes the probability that the user will like content $c$. We call this the {\em relevance score} of content $c$.

The second case we consider corresponds to the scenario when the characteristics of the users and content are {\em dynamic}.
For online multimedia content, especially for social media, it is known that both the user and the content characteristics are dynamic and noisy \cite{naaman2012social}, hence the problem exhibits concept drift \cite{minku2010impact}. Formally, a concept is the distribution of the problem, i.e., the joint distribution of the user and content characteristics, at a certain point of time \cite{narasimhamurthy2007framework}. Concept drift is a change in this distribution. 
For the case with concept drift, we propose a learning algorithm that takes into account the speed of the drift to decide what window of past observations to use in estimating the relevance score. The proposed learning algorithm has theoretical performance guarantees in contrast to prior work on concept drift which mainly deal with the problem in a ad-hoc manner. 
Indeed, it is customary to assume that online content is highly dynamic. A certain type of content may become popular for a certain period of time, and then, its popularity may decrease over time and a new content may emerge as popular. In addition, although the type of the content remains the same, such as soccer news, its popularity may change over time due to exogenous events such as the World Cup etc. Similarly, a certain type of content may become popular for a certain type of demographics (e.g., users of a particular age, gender, profession, etc.). 
However, over time the interest of these users may shift to other types of content. In such cases, where the popularity of content changes over time for a user with context $x$, $\pi_{c}(x,t)$ denotes the probability that the user at time $t$ will like content $c$. 

\comment{
\begin{figure}
\begin{center}
\includegraphics[width=0.95\columnwidth]{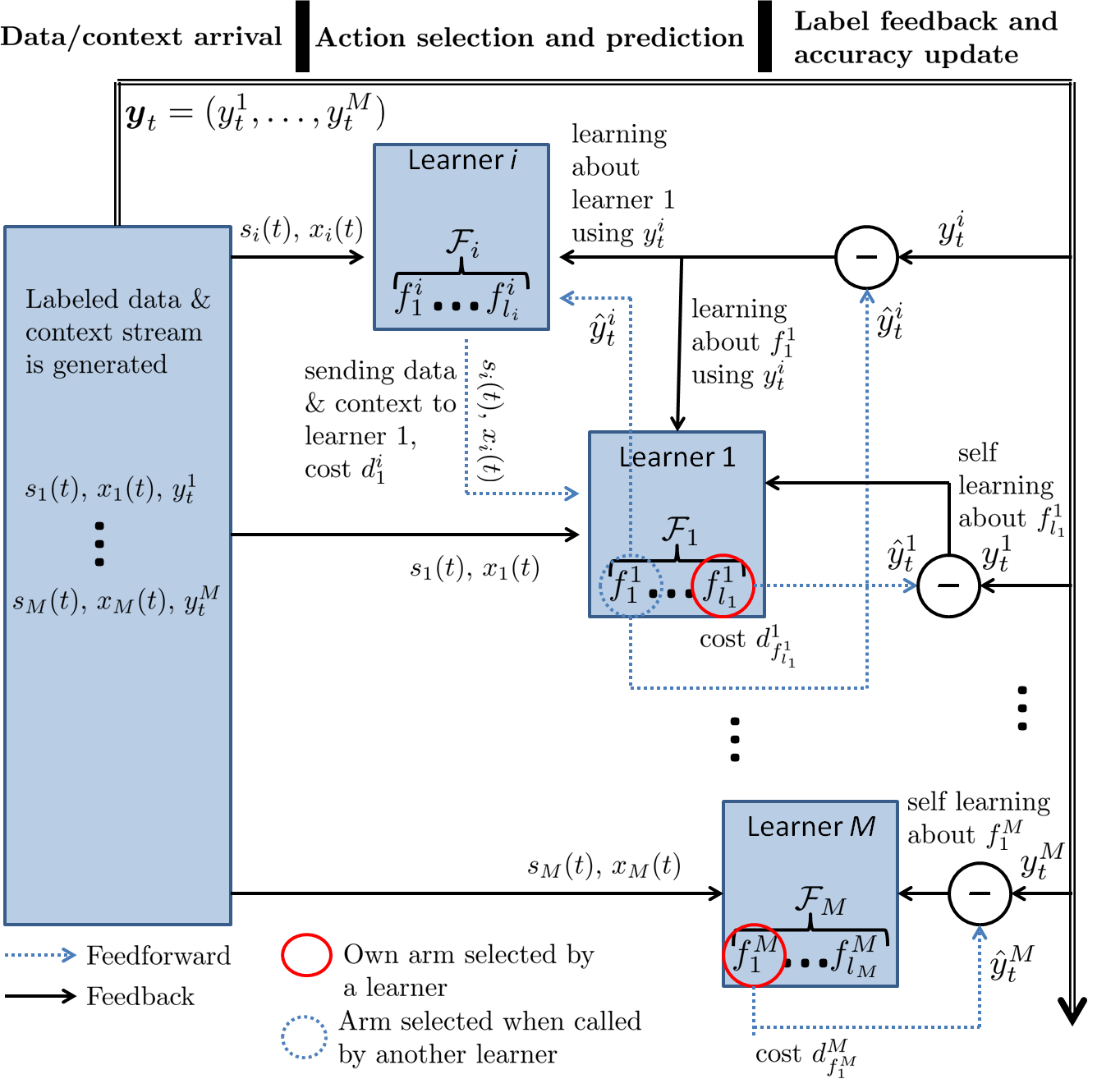}
\vspace{-0.1in}
\caption{Operation of the distributed data classification system during a time slot.} 
%\vspace{-0.4in}
\label{fig:system}
\end{center}
\vspace{-0.2in}
\end{figure}
}

As we stated earlier, a CA $i$ can either recommend content from multimedia sources that it is directly connected to or can ask another CA for content. By asking for content $c$ from another CA $j$, CA $i$ will incur cost $d^i_j \geq 0$. 
\newc{For the purpose of our paper, the cost is a {\em generic} term. 
For instance, it can be a payment made to CA $j$ to display it to CA $i$'s user, or it may be associated with the advertising loss CA $i$ incurs by directing its user to CA $j$'s website. When the cost is payment, it can be money, tokens \cite{xu2013efficient} or Bitcoins \cite{nakamoto2008bitcoin}.}
Since this cost is bounded, without loss of generality we assume that $d^i_j \in [0,1]$ for all $i,j \in {\cal M}$.
In order make our model general, we also assume that there is a cost associated with recommending a type of content $c \in {\cal C}_i$, which is given by $d^i_c \in [0,1]$, for CA $i$. For instance, this can be a payment made to the multimedia source that owns content $c$. 

An intrinsic assumption we make is that the CAs are cooperative. That is, CA $j \in {\cal M}_{-i}$ will return the content that is mostly to be liked by CA $i$'s user when asked by CA $i$ to recommend a content. This cooperative structure can be justified as follows. Whenever a user likes the content of CA $j$ (either its own user or user of another CA), CA $j$ obtains a benefit. This can be either an additional payment made by CA $i$ when the content recommended by CA $j$ is liked by CA $i$'s user, or it can simply be the case that whenever a content of CA $j$ is liked by someone its popularity increases. 
However, we assume that the CAs' decisions do not change their pool of users. The future user arrivals to the CAs are independent of their past content matching strategies. For instance, users of a CA may have monthly or yearly subscriptions, so they will not shift from one CA to another CA when they like the content of the other CA. 

The goal of CA $i$ is to explore the matching actions in ${\cal K}_i$ to learn the 
best content for each context, while at the same time exploiting the best content for the user with context $x_i(t)$ arriving at each time instance $t$ to maximize its total number of likes minus costs. 
CA $i$'s problem can be modeled as a contextual bandit problem \cite{slivkins2009contextual, dudik2011efficient, langford2007epoch, chu2011contextual}, where likes and costs translate into rewards. 
In the next subsection, we formally define the benchmark solution which is computed using perfect knowledge about the probability that a content $c$ will be liked by a user with context $x$ (which requires complete knowledge of user and content characteristics). Then, we define the regret which is the performance loss due to uncertainty about the user and content characteristics. 

\vspace{-0.1in}
\subsection{Optimal Content Matching with Complete Information} \label{sec:centralized}

Our benchmark when evaluating the performance of the learning algorithms is the optimal solution which always recommends the content with the highest relevance score minus cost for CA $i$ from the set ${\cal C}$ given context $x_i(t)$ at time $t$. This corresponds to selecting the best matching action in ${\cal K}_i$ given $x_i(t)$.
Next, we define the expected rewards of the matching actions, and the {\em action selection} policy of the benchmark.
For a matching action $k \in {\cal M}_{-i}$, its relevance score is given as $\pi_k(x) :=  \pi_{c^*_k(x)}(x)$, where $c^*_k(x) := \argmax_{c \in {\cal C}_j} \pi_c(x)$.
For a matching action $k \in {\cal C}_i$ its relevance score is equal to the relevance score of content $k$.
The expected reward of CA $i$ from choosing action $k \in {\cal K}_i$ is given by the quasilinear utility function
\begin{align}
\mu^i_k(x) := \pi_k(x) - d^i_k  \label{eqn:netreward}
\end{align}
where $d^i_k \in [0,1]$ is the normalized cost of choosing action $k$ for CA $i$.
Our proposed system will also work for more general expected reward functions as long as the expected reward of a learner is a function of the relevance score of the chosen action and the cost (payment, communication cost, etc.) associated with choosing that action.
The {\em oracle} benchmark is given by
\begin{align}
k_i^*(x) := \argmax_{k \in {\cal K}_i} \mu^i_k(x) ~~ \forall x \in {\cal X}. \label{eqn:opt2}
\end{align}
The oracle benchmark depends on relevance scores as well as costs of matching content from its own content network or other CA's content network. 
The case $d^i_k = 0$ for all $k \in {\cal K}_i$ and $i \in {\cal M}$, corresponds to the scheme in which content matching has zero cost, hence
%\begin{align}
$k^*_i(x) = \argmax_{k \in {\cal K}_i} \pi_k(x) = \argmax_{c \in {\cal C}} \pi_c(x)$.  
%\end{align}  
This corresponds to the best centralized solution, where CAs act as a single entity.
On the other hand, when $d^i_k \geq 1$ for all $k \in {\cal M}_{-i}$ and $i \in {\cal M}$, in the oracle benchmark a CA must not cooperate with any other CA and should only use its own content. Hence
%\begin{align}
$k^*_i(x) = \argmax_{c \in {\cal C}_i} (\pi_c(x) - d^i_c)$.    
%\end{align}
%
%distributed solution which is computed offline given the complete knowledge about the data arrival process and classifier accuracies. 
%For simplicity, we focus on the distributed classification problem of classifier $i$ 
%In the optimal classification scheme, for each context $x \in {\cal X}$ classifier $i$ wants to find the optimal classifier from the set $\cup_{i \in {\cal M}} {\cal F}_i$, i.e., the classifier $k$ with the highest accuracy minus the cost of classification $d_k$.
%
%
%\com{In addition to Assumption \ref{}, we can prove tighter regret bounds if we assume more structure on the optimal solution. Specifically, the following assumption states that for almost all the points in the context space, there is a  gap between the accuracies of optimal and suboptimal classifiers. (Since the diameters of the cubes decrease with the number of slices, when we increase $m_T$, the suboptimality gap for the cubes adjacent to the boundary cubes decreases since they will become closer to the boundary. Therefore even if we assume that the optimal scheme partitions ${\cal X}$ into  $|\cup_{i \in {\cal M}} {\cal F}_i|$ sets in each of which a single classifier is optimal, the volume of cubes for which the suboptimality is less than $\epsilon>0$ remains constant. Therefore, if we do not increase the rate of explorations regret due to suboptimalities will be linear in time).}
%
% 
%Note that partitioning $[0,1]^d$ is not enough because the noise can carry a data point outside of $[0,1]^d$.
%
In the following subsections, we will show that independent of the values of relevance scores and costs, our algorithms will achieve sublinear regret (in the number of users or equivalently time) with respect to the oracle benchmark.

%
%\add{\vspace{-0.2in}}
\vspace{-0.1in}
\subsection{The Regret of Learning}

In this subsection we define the regret as a performance measure of the learning algorithm used by the CAs. Simply, the regret is the loss incurred due to the unknown system dynamics. Regret of a learning algorithm which selects the matching action/arm $a_i(t)$ at time $t$ for CA $i$ is defined with respect to the best matching action $k_i^*(x)$ given in (\ref{eqn:opt2}). 
Then, the regret of CA $i$ at time $T$ is
\add{\vspace{-0.05in}}
\begin{align}
%\textrm{({\bf P2})  }
R_i(T) &:= \sum_{t=1}^T \left( \pi_{k_i^*(x_i(t))}(x_i(t)) - d^i_{k_i^*(x_i(t))} \right) \notag \\ 
&- \textrm{E} \left[ \sum_{t=1}^T \left( \mathrm{I} ( y_i(t) = L ) - d^i_{a_i(t)} \right) \right] \label{eqn:tmmregret}.
\end{align}
Regret gives the convergence rate of the total expected reward of the learning algorithm to the value of the optimal solution given in (\ref{eqn:opt2}). Any algorithm whose regret is sublinear, i.e., $R_i(T) = O(T^\gamma)$ such that $\gamma<1$, will converge to the optimal solution in terms of the average reward.

\newc{
A summary of notations is given in Table \ref{tab:notation1}. 
In the next section, we propose an online learning algorithm which achieves sublinear regret when the user and content characteristics are static. 
}

\newc{
\begin{table}
\centering
{
{\fontsize{9}{10}\selectfont
\begin{tabular}{|l|}
\hline
${\cal M}$: Set of all CAs \\
\hline
${\cal C}_i$: Contents in the Content Network of CA $i$ \\ 
\hline 
$C_{\max}$: $\max_{i \in {\cal M}} |{\cal C}_i|$ \\
\hline
${\cal C}$: Set of all contents  \\
\hline
${\cal X}=[0,1]^d$: Context space \\
\hline
${\cal Y}$: Set of feedbacks a user can give \\
\hline
$x_i(t)$: $d$-dimensional context of $t$th user of CA $i$ \\
\hline
$y_i(t)$: Feedback of the $t$th user of CA $i$ \\
\hline
${\cal K}_i$: Set of content matching actions of CA $i$ \\
\hline 
$\pi_c(x)$: Relevance score of content $c$ for context $x$ \\
\hline
$d^i_k$: Cost of choosing matching action $k$ for CA $i$ \\
\hline
$\mu^i_k(x)$: Expected reward (static) of CA $i$ from  \\ matching action $k$ for context $x$ \\
\hline
$k^*_i(x)$: Optimal matching action of CA $i$ given \\ context $x$ (oracle benchmark) \\
\hline
$R_i(T)$: Regret of CA $i$ at time $T$ \\
\hline
$\beta_a := \sum_{t=1}^{\infty} 1/t^a$ \\
\hline
\end{tabular}
}
}
\caption{Notations used in problem formulation.}
\vspace{-0.3in}
\label{tab:notation1}
\end{table}
}

%which divides the context space $[0,1]^d$ to $(m_T)^d$ h

%ypercubes, and estimates the best classifier or best learner to call, in each of these hypercubes. Here, the number $m_T$ depends on the time horizon $T$ and is nondecreasing in $T$ which means that the number of hypercubes we consider is nondecreasing in $T$. The longer the time horizon, the finer should the partitions be in order to control the suboptimality resulting from taking averages over the entire hypercube. Secondly, we propose a distributed zooming algorithm that adaptively adjusts the number of hypercubes by zooming into the regions of the context space with high context arrival intensity. 

%\input{centralized}
\vspace{-0.1in}
\section{A Distributed Online Content Matching Algorithm} \label{sec:iid}

In this section we propose an online learning algorithm for content matching when the user and content characteristics are static.
In contrast to prior online learning algorithms that exploit the context information \cite{hazan2007online, slivkins2009contextual, dudik2011efficient, langford2007epoch, chu2011contextual, lu2010contextual}, which consider a single learner setting, the proposed algorithm helps a CA to learn from the {\em experience} of other CAs. 
With this mechanism, a CA is able to recommend content from multimedia sources that it has no direct connection, without needing to know the IDs of such multimedia sources and their content. It learns about these multimedia sources only through the other CAs that it is connected to. 

In order to bound the regret of this algorithm analytically we use the following assumption. When the content characteristics are static, 
we assume that each type of content has similar relevance scores for similar contexts; we formalize this in terms of a Lipschitz condition.
\begin{assumption} \label{ass:lipschitz2}
There exists $L>0$, $\gamma>0$ such that for all $x,x' \in {\cal X}$ and $c \in {\cal C}$, we have
$|\pi_{c}(x) - \pi_{c}(x')| \leq L ||x-x'||^\gamma$.    
\end{assumption}

Assumption \ref{ass:lipschitz2} indicates that the probability that a type $c$ content is liked by 
users with similar contexts will be similar to each other. For instance, if two users have similar age, gender, etc., then it is more likely that they like the same content. 
We call $L$ the {\em similarity constant} and $\gamma$ the {\em similarity exponent}. These parameters will depend on the characteristics of the users and the content. We assume that $\gamma$ is known by the CAs. However, an unknown $\gamma$ can be estimated online using the history of likes and dislikes by users with different contexts, and our proposed algorithms can be modified to include the estimation of $\gamma$. 
 
In view of this assumption, the important question becomes how to learn from the past experience which content to match with the current user. 
We answer this question by proposing an algorithm which partitions the context space of a CA, and learns the relevance scores of different types of content for each set in the partition, based only on the past experience in that set. 
The algorithm is designed in a way to achieve optimal tradeoff between the size of the partition and the past observations that can be used together to learn the relevance scores. It also includes a mechanism to help CAs learn from each other's users. 
We call our proposed algorithm the {\em DIStributed COntent Matching algorithm} (DISCOM), and its pseudocode is given in  Fig. \ref{fig:CLUP}, Fig. \ref{fig:CLUPmax} and Fig. \ref{fig:CLUPcoop}.

Each CA $i$ has two tasks: matching content with its own users and matching content with the users of other CAs when requested by those CAs.
We call the first task the {\em maximization task} (implemented by $\textrm{DISCOM}_{\textrm{max}}$ given in Fig. \ref{fig:CLUPmax}), since the goal of CA $i$ is to maximize the number of likes from its own users. The second task is called the {\em cooperation task}  (implemented by $\textrm{DISCOM}_{\textrm{coop}}$ given in Fig. \ref{fig:CLUPcoop}), since the goal of CA $i$ is to help other CAs obtain content from its own content network in order to maximize the likes they receive from their users.
This cooperation is beneficial to CA $i$ because of numerous reasons. Firstly, since every CA cooperates, CA $i$ can reach a much larger set of content including the content from other CA's content networks, hence will be able to provide content with higher relevance score to its users. Secondly, when CA $i$ helps CA $j$, it will observe the feedback of CA $j$'s user for the matched content, hence will be able to update the estimated relevance score of its content, which is beneficial if a user similar to CA $j$'s user arrives to CA $i$ in the future. 
Thirdly, payment mechanisms can be incorporated to the system such that CA $i$ gets a payment from CA $j$ when its content is liked by CA $j$'s user. 

In summary, there are two types of content matching actions for a user of CA $i$. In the first type, the content is recommended from a source that is directly connected to CA $i$, while in the second type, the content is recommended from a source that CA $i$ is connected through another CA.
The information exchange between multimedia sources and CAs for these two types of actions is shown in Fig. \ref{fig:twodiffrec1} and Fig. \ref{fig:twodiffrec2}. 

\begin{figure}
\begin{center}
\includegraphics[width=0.9\columnwidth]{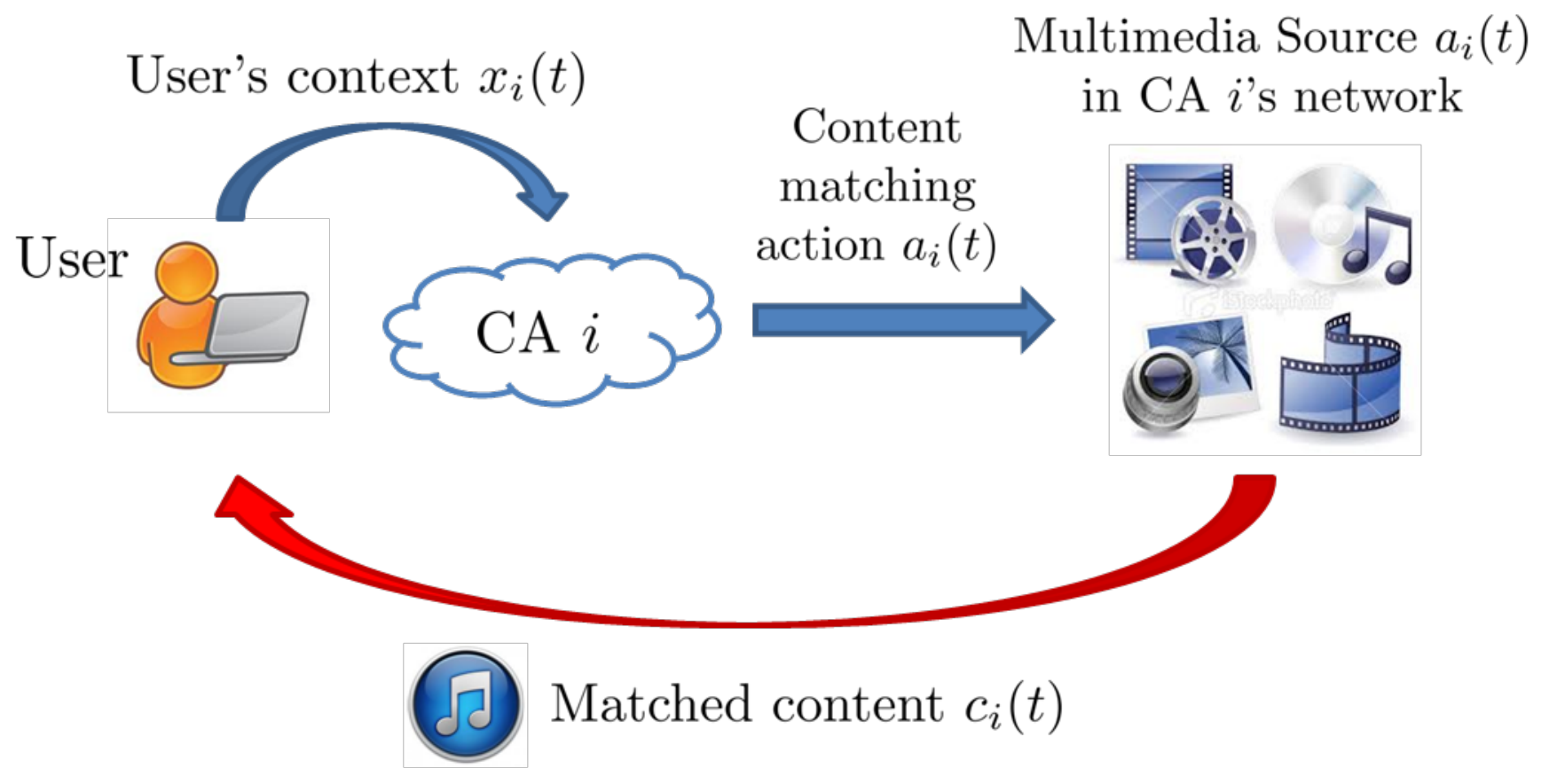}
\vspace{-0.1in}
\caption{Content matching within own content network.} 
%\vspace{-0.4in}
\label{fig:twodiffrec1}
\end{center}
\vspace{-0.2in}
\end{figure}

\begin{figure}
\begin{center}
\includegraphics[width=\columnwidth]{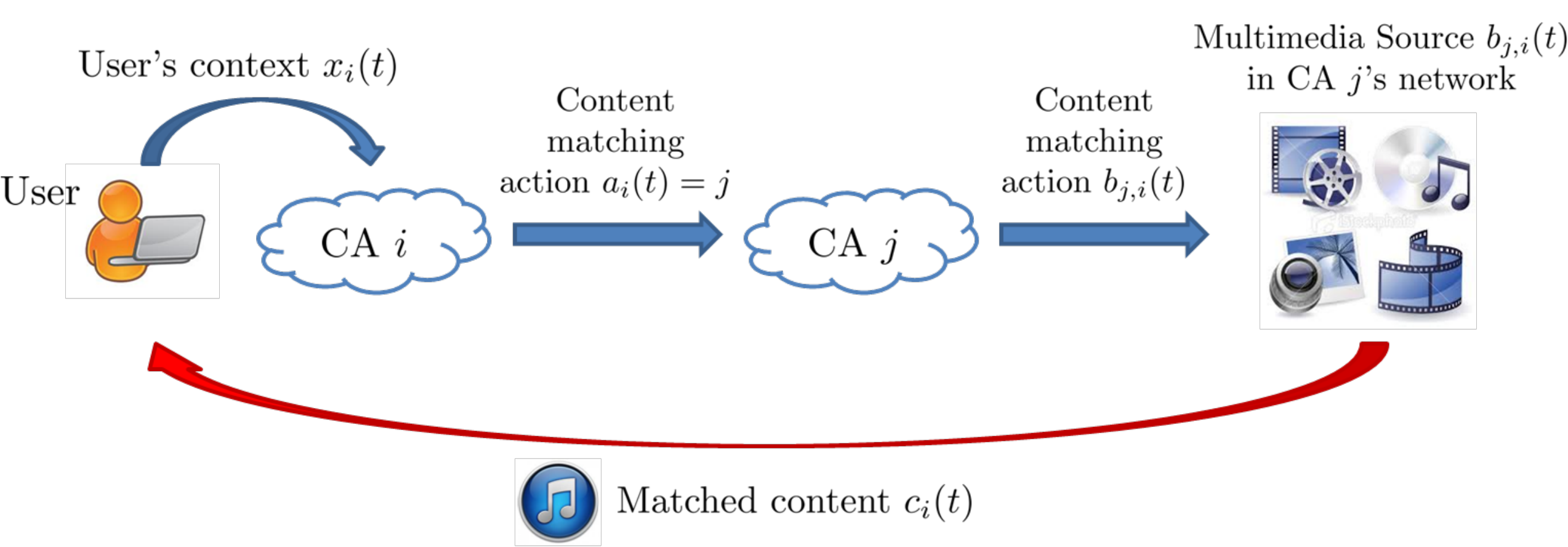}
\vspace{-0.2in}
\caption{Content matching from the network of another CA.} 
%\vspace{-0.4in}
\label{fig:twodiffrec2}
\end{center}
\vspace{-0.2in}
\end{figure}

Let $T$ be the time horizon of interest (equivalent to the number of users that arrive to each CA).
DISCOM creates a partition of ${\cal X} = [0,1]^d$ based on $T$. For instance $T$ can be the average number of visits to the CA's website in one day. Although in reality the average number of visits to different CAs can be different, our analysis of the regret in this section will hold since it is the worst-case analysis (assuming that users arrive only to CA $i$, while the other CAs only learn through CA $i$'s users). 
Moreover, the case of heterogeneous number of visits can be easily addressed if each CA informs other CAs about its average number of visits. Then, CA $i$ can keep $M$ different partitions of the context space; one for itself and $M-1$ for the other CAs. If called by CA $j$, it will match a content to CA $j$'s user based on the partition it keeps for CA $j$. Hence, we focus on the case when $T$ is common to all CAs. 

We first define $m_T$ as the {\em slicing level} used by DISCOM, which is an integer that is used to partition ${\cal X}$.
DISCOM forms a partition of ${\cal X}$ consisting of $(m_T)^d$ sets (hypercubes) where each set is a $d$-dimensional hypercube with edge length $1/m_T$. This partition is denoted by ${\cal P}_T$. 
The hypercubes in ${\cal P}_T$ are oriented such that one of them has a corner located at the origin of the $d$-dimensional Euclidian space. 
It is clear that the number of hypercubes is increasing in $m_T$, while their size is decreasing in $m_T$. When $m_T$ is small each hypercube covers a large set of contexts, hence the number of past observations that can be used to estimate relevance scores of matching actions in each set is large. However, the variation of the true relevance scores of the contexts within a hypercube increases with the size of the hypercube. 
DISCOM should set $m_T$ to a value that balances this tradeoff. 

A hypercube in ${\cal P}_T$ is denoted by $p$. The hypercube in ${\cal P}_T$ that contains context $x_i(t)$ is denoted by $p_i(t)$. When $x_i(t)$ is located at a boundary of multiple hypercubes in ${\cal P}_T$, it is randomly assigned to one of these hypercubes.

\begin{figure}[htb]
\add{\vspace{-0.1in}}
\fbox {
\begin{minipage}{0.95\columnwidth}
{\fontsize{9}{9}\selectfont
\flushleft{DISCOM for CA $i$:}
\begin{algorithmic}[1]
\STATE{Input: $H_1(t)$, $H_2(t)$, $H_3(t)$, $T$, $m_T$}
\STATE{Initialize: Partition ${\cal X}$ into hypercubes denoted by ${\cal P}_T$}
\STATE{Initialize: Set counters
$N^i_p =0$, $\forall p \in {\cal P}_T$,
$N^i_{k,p}=0, \forall k \in {\cal K}_i, p \in {\cal P}_T$, $N^{\textrm{tr}, i}_{j, p} =0, \forall j \in {\cal M}_{-i}, p \in {\cal P}_T$
}
\STATE{Initialize: Set relevance score estimates $\bar{r}^i_{k,p} =0$, $\forall k \in {\cal K}_i$, $p \in {\cal P}_T$
}
\WHILE{$t \geq 1$}
\STATE{Run $\textrm{DISCOM}_{\textrm{max}}$ to find $p=p_i(t)$, to obtain a matching action $a_i$, and value of $train$ flag}
\STATE{If $a_i \in {\cal M}_{-i}$ ask CA $a_i$ for content and pass $x_i(t)$}
\STATE{Receive ${\cal CA}_i(t)$, the set of CAs who requested content from CA $i$, and their contexts}
\IF{${\cal CA}_i(t) \neq \emptyset$}
\STATE{Run $\textrm{DISCOM}_{\textrm{coop}}$ to obtain the content to be selected $\boldsymbol{b}_i := \{ b_{i,j}  \}_{j \in  {\cal CA}_i(t)}$ and hypercubes that the contexts of the users in ${\cal CA}_i(t)$ lie in $\boldsymbol{p}_i := \{ p_{i,j}  \}_{j \in  {\cal CA}_i(t)}$}
\ENDIF
\IF{$a_i \in {\cal C}_i$}
\STATE{\newc{Pay cost $d^i_{a_i}$, obtain content $a_i$}}
\STATE{\newc{Show $a_i$ to the user}, receive feedback $r \in \{0,1\}$ drawn from $\textrm{Ber}_{a_i}(x_i(t))$\footnote{$\textrm{Ber}_{a_i}(x_i(t))$ is the Bernoulli distribution with expected value $\pi_{a_i}(x_i(t))$}}
\ELSE
\STATE{\newc{Pay cost $d^i_{a_i}$, obtain content $b_{a_i,i}$ from CA $a_i$}}
\STATE{\newc{Show $b_{a_i,i}$ to the user}, receive feedback $r \in \{0,1\}$ drawn from $\textrm{Ber}_{b_{a_i,i}}(x_i(t))$}
\ENDIF
\IF{$train = 1$}
\STATE{$N^{\textrm{tr}, i}_{a_i, p} ++$}
\ELSE
\STATE{$\bar{r}^i_{a_i,p} = (\bar{r}^i_{a_i,p} N^i_{a_i,p}  + r)/(N^i_{a_i,p}+1)$}
\STATE{$N^i_{p}++$, $N^i_{a_i,p} ++$}
\ENDIF
\IF{${\cal CA}_i(t) \neq \emptyset$}
\FOR{$j \in {\cal CA}_i(t)$}
\STATE{\newc{Send content $b_{i,j}$ to CA $j$'s user}}
\STATE{Observe feedback $r$ drawn from $\textrm{Ber}_{b_{i,j}}(x_j(t))$}
\STATE{$\bar{r}^i_{b_{i,j},p_{i,j}} = \frac{\bar{r}^i_{b_{i,j},p_{i,j}} N^i_{b_{i,j},p_{i,j}}  + r}{N^i_{b_{i,j},p_{i,j}}+1}$}
\STATE{$N^i_{p_{i,j}}++$, $N^i_{b_{i,j},p_{i,j}}++$}
\ENDFOR
\ENDIF
\STATE{$t=t+1$}
\ENDWHILE
\end{algorithmic}
}
\end{minipage}
} \caption{Pseudocode for DISCOM algorithm.} \label{fig:CLUP}
\add{\vspace{-0.12in}}
\end{figure}

\begin{figure}[htb]
\add{\vspace{-0.1in}}
\fbox {
\begin{minipage}{0.95\columnwidth}
{\fontsize{9}{7}\selectfont
\flushleft{$\textrm{DISCOM}_{\max}$ (maximization part of DISCOM) for CA $i$:}
\begin{algorithmic}[1]
\STATE{$train=0$}
\STATE{Find the hypercube in ${\cal P}_T$ that $x_i(t)$ belongs to, i.e., $p_i(t)$}
\STATE{Let $p = p_i(t)$}
\STATE{Compute the set of under-explored matching actions ${\cal C}^{\textrm{ue}}_{i,p}(t)$ given in (\ref{eqn:underexploreda})}
\IF{${\cal C}^{\textrm{ue}}_{i,p}(t) \neq \emptyset$}
\STATE{Select $a_i$ randomly from ${\cal C}^{\textrm{ue}}_{i,p}(t)$}
\ELSE
\STATE{Compute the set of training candidates ${\cal M}^{\textrm{ct}}_{i,p}(t)$ given in (\ref{eqn:traincand})}
\STATE{//Update the counters of training candidates}
\FOR{$j \in {\cal M}^{\textrm{ut}}_{i,p}(t)$}
\STATE{Obtain $N^j_p$ from CA $j$, set $N^{\textrm{tr}, i}_{j,p} = N^j_p - N^i_{j,p}$}
\ENDFOR
\STATE{Compute the set of under-trained CAs ${\cal M}^{\textrm{ut}}_{i,p}(t)$ given in (\ref{eqn:undertrained})}
\STATE{Compute the set of under-explored CAs ${\cal M}^{\textrm{ue}}_{i,p}(t)$ given in (\ref{eqn:underexploredl})}
\IF{${\cal M}^{\textrm{ut}}_{i,p}(t) \neq \emptyset$}
\STATE{Select $a_i$ randomly from ${\cal M}^{\textrm{ut}}_{i,p}(t)$, $train=1$}
%\STATE{Run {\bf Train}($i$, $j$, $N^i_{1,j,l}$).}
\ELSIF{${\cal M}^{\textrm{ue}}_{i,p}(t) \neq \emptyset$}
\STATE{Select $a_i$ randomly from ${\cal M}^{\textrm{ue}}_{i,p}(t)$}
\ %STATE{Run {\bf Explore}($i$, $j$, $N^i_{2,j,l}$, $\bar{r}^i_{j,l}$)}
\ELSE
\STATE{Select $a_i$ randomly from $\argmax_{k \in {\cal K}_i} \bar{r}^i_{k,p} - d^i_k$}
\ENDIF
\ENDIF
\end{algorithmic}
}
\end{minipage}
} \caption{Pseudocode for the maximization part of DISCOM algorithm.} \label{fig:CLUPmax}
\end{figure}

\begin{figure}[htb]
\add{\vspace{-0.1in}}
\fbox {
\begin{minipage}{0.95\columnwidth}
{\fontsize{9}{7}\selectfont
\flushleft{$\textrm{DISCOM}_{\textrm{coop}}$ (cooperation part of DISCOM) for CA $i$}
\begin{algorithmic}[1]
\FOR{$j \in {\cal CA}_i(t)$}
\STATE{Find the set in ${\cal P}_T$ that $x_j(t)$ belongs to, i.e., $p_{i,j}$}
\STATE{Compute the set of under-explored matching actions ${\cal C}^{\textrm{ue}}_{i,p_{i,j}}(t)$ given in (\ref{eqn:underexploreda})}
\IF{${\cal C}^{\textrm{ue}}_{i,p_{i,j}}(t) \neq \emptyset$}
\STATE{Select $b_{i,j}$ randomly from ${\cal C}^{\textrm{ue}}_{i,p_{i,j}}(t)$}
\ELSE
\STATE{$b_{i,j} = \argmax_{c \in {\cal C}_i} \bar{r}^i_{c, p_{i,j}}$}
\ENDIF
\ENDFOR
\end{algorithmic}
}
\end{minipage}
} \caption{Pseudocode for the cooperation part of DISCOM algorithm.} \label{fig:CLUPcoop}
\add{\vspace{-0.12in}}
\end{figure}

$\textrm{DISCOM}_{\textrm{max}}$ operates as follows.
CA $i$ matches its user at time $t$ with a content by taking a matching action based on one of the three phases: {\em training} phase in which CA $i$ requests content from another CA $j$ for the purpose of helping CA $j$ to learn the relevance score of content in its content network for users with context $x_i(t)$ (but CA $i$ does not update the relevance score for CA $j$ because it thinks that CA $j$ may not know much about its own content), 
{\em exploration} phase in which CA $i$ selects a matching action in ${\cal K}_i$ and updates its relevance score based on the feedback of its user,
and {\em exploitation} phase in which CA $i$ chooses the matching action with the highest relevance score minus cost. 

Since the CAs are cooperative, when another CA requests content from CA $i$, CA $i$ will choose content from its content network with the highest estimated relevance score for the user of the requesting CA. 
To maximize the number of likes minus costs in exploitations,
CA $i$ must have an accurate estimate of the relevance scores of other CAs. This task is not trivial since CA $i$ does not know the content network of other CAs. 
In order to do this, CA $i$ should smartly select which of its users' feedbacks to use when estimating the relevance score of CA $j$. 
The feedbacks should come from previous times at which CA $i$ has a very high confidence that the content of CA $j$ matched with its user is the one with the highest relevance score for the context of CA $i$'s user.
Thus, the training phase of CA $i$ helps other CAs build accurate estimates about the relevance scores of their content, before CA $i$ uses any feedback for content coming from these CAs to form relevance score estimates about them.
In contrast, the exploration phase of CA $i$ helps it to build accurate estimates about the relevance score of its matching actions. 

At time $t$, the phase that CA $i$ will be in is determined by the amount of time it had explored, exploited or trained for past users with contexts similar to the context of the current user. 
For this CA $i$ keeps counters and control functions which are described below.
Let $N^i_p(t)$ be the number of user arrivals to CA $i$ with contexts in $p \in {\cal P}_T$ by time $t$ (its own arrivals and arrivals to other CAs who requested content from CA $i$) except the training phases of CA $i$.
For $c \in {\cal C}_i$, let $N^i_{c,p}(t)$ be the number of times content $c$ is selected in response to a user arriving to CA $i$ with context in hypercube $p$ by time $t$ (including times other CAs request content from CA $i$ for their users with contexts in set $p$).
Other than these, CA $i$ keeps two counters for each other CA in each set in the partition, which it uses to decide the phase it should be in.
The first one, i.e., $N^{\textrm{tr}, i}_{j,p}(t)$, is an estimate on the number of user arrivals with contexts in $p$ to CA $j$ from all CAs except the training phases of CA $j$ and exploration, exploitation phases of CA $i$.
This counter is only updated when CA $i$ thinks that CA $j$ should be trained.
The second one, i.e., $N^i_{j,p}(t)$, counts the number of users of CA $i$ with contexts in $p$ for which content is requested from CA $j$ at exploration and exploitation phases of CA $i$ by time $t$.

At each time slot $t$, CA $i$ first identifies $p_i(t)$.
Then, it chooses its phase at time $t$ by giving highest priority to exploration of content in its own content network, second highest priority to training of the other CAs, third highest priority to exploration of the other CAs, and lowest priority to exploitation. 
The reason that exploration of own content has a higher priority than training of other CAs is that it will minimize the number of times CA $i$ will be trained by other CAs, which we describe below.

First, CA $i$ identifies the set of under-explored content in its content network:
\begin{align}
{\cal C}^{\textrm{ue}}_{i,p}(t) &:= \{ c \in {\cal C}_{i}: N^i_{c,p}(t) \leq H_1(t) \}
\label{eqn:underexploreda}
\end{align}
where $H_1(t)$ is a deterministic, increasing function of $t$ which is called {\em the control function}.
The value of this function will affect the regret of DISCOM.
For $c \in {\cal C}_i$, the accuracy of relevance score estimates increase with $H_1(t)$, hence it should be selected to balance the tradeoff between accuracy and the number of explorations.
If ${\cal C}^{\textrm{ue}}_{i,p}(t)$ is non-empty, CA $i$ enters the exploration phase and randomly selects a content in this set to explore. 
Otherwise, it identifies the set of training candidates:
\begin{align}
{\cal M}^{\textrm{ct}}_{i,p}(t) := \{ j \in {\cal M}_{-i}: N^{\textrm{tr}, i}_{j,p}(t) \leq H_2(t)\} \label{eqn:traincand}
\end{align}
where $H_2(t)$ is a control function similar to $H_1(t)$.
Accuracy of other CA's relevance score estimates of content in their own networks increases with $H_2(t)$, hence it should be selected to balance the possible reward gain of CA $i$ due to this increase with the reward loss of CA $i$ due to the number of trainings.
If this set is non-empty, CA $i$ asks the CAs $j \in {\cal M}^{\textrm{ct}}_{i,p}(t)$ to report $N^j_p(t)$.
Based in the reported values it recomputes $N^{\textrm{tr},i}_{j,p}(t)$ as $N^{\textrm{tr},i}_{j,p}(t) = N^j_p(t) - N^i_{j,p}(t)$.
Using the updated values, CA $i$ identifies the set of under-trained CAs:
\begin{align}
{\cal M}^{\textrm{ut}}_{i,p}(t) := \{ j \in {\cal M}_{-i}: N^{\textrm{tr}, i}_{j,p}(t) \leq H_2(t)\}. \label{eqn:undertrained}
\end{align}
If this set is non-empty, CA $i$ enters the training phase and randomly selects a CA in this set to train it.
%\footnote{Most of the regret bounds proposed in this paper can also be achieved by setting $N^{\textrm{tr}, i}_{j,p}(t)$ to be the number of times learner $i$ trains learner $j$ by time $t$, without considering other context observations of learner $j$. 
%However, by recomputing $N^{\textrm{tr}, i}_{j,p}(t)$, learner $i$ can avoid many unnecessary trainings especially when own context arrivals of learner $j$ is adequate for it to form accurate estimates about its arms for  set $p$ or when learners other than $i$ have already helped learner $j$ to build accurate estimates for its arms in set $p$.}
When ${\cal M}^{\textrm{ct}}_{i,p}(t)$ or ${\cal M}^{\textrm{ut}}_{i,p}(t)$ is empty, this implies that there is no under-trained CA, hence CA $i$ checks if there is an under-explored matching action. 
The set of CAs for which CA $i$ does not have accurate relevance scores is given by 
\begin{align}
{\cal M}^{\textrm{ue}}_{i,p}(t) &:=  \{ j \in {\cal M}_{-i}: N^i_{j,p}(t) \leq H_3(t) \}
\label{eqn:underexploredl}
\end{align}
where $H_3(t)$ is also a control function similar to $H_1(t)$.
If this set is non-empty, CA $i$ enters the exploration phase and randomly selects a CA in this set to request content from to explore it.
Otherwise, CA $i$ enters the exploitation phase in which it selects the matching action with the highest estimated relevance score minus cost for its user with context $x_i(t) \in p = p_i(t)$, i.e.,
\begin{align}
a_i(t) \in \argmax_{k \in {\cal K}_i} \bar{r}^i_{k,p}(t) - d^i_k \label{eqn:maximizer}
\end{align}
where $\bar{r}^i_{k,p}(t)$ is the sample mean estimate of the relevance score of CA $i$ for matching action $k$ at time $t$, which is computed as follows. 
For $j \in {\cal M}_{-i}$, let ${\cal E}^i_{j,p}(t)$ be the set of feedbacks collected by CA $i$ at times it selected CA $j$ while CA $i$'s users' contexts are in set $p$ in its exploration and exploitation phases by time $t$.
For estimating the relevance score of contents in its own content network, CA $i$ can also use the feedback obtained from other CAs' users at times they requested content from CA $i$.
In order to take this into account, for $c \in {\cal C}_i$, let ${\cal E}^i_{c,p}(t)$ be the set of feedbacks observed by CA $i$ at times it selected its content $c$ for its own users with contexts in set $p$ union the set of feedbacks observed by CA $i$ when it selected its content $c$ for the users of other CAs with contexts in set $p$ who requests content from CA $i$ by time $t$.

Therefore, sample mean relevance score of matching action $k \in {\cal K}_i$ for users with contexts in set $p$ for CA $i$ is defined as
%
%\begin{align*}
$\bar{r}^i_{k,p}(t) = \left(\sum_{r \in {\cal E}^i_{k,p}(t)} r \right)/|{\cal E}^i_{k,p}(t)|$,
%\end{align*}
%
An important observation is that computation of $\bar{r}^i_{k,p}(t)$ does not take into account the matching costs.
Let $\hat{\mu}^i_{k,p}(t) := \bar{r}^i_{k,p}(t) - d^i_k$ be the estimated {\em net reward} (relevance score minus cost) of matching action $k$ for set $p$.
Of note, when there is more than one maximizer of (\ref{eqn:maximizer}), one of them is randomly selected.
In order to run DISCOM, CA $i$ does not need to keep the sets ${\cal E}^i_{k,p}(t)$ in its memory. 
$\bar{r}^i_{k,p}(t)$ can be computed by using only $\bar{r}^i_{k,p}(t-1)$ and the feedback at time $t$.

The cooperation part of DISCOM, i.e., $\textrm{DISCOM}_{\textrm{coop}}$ operates as follows.
Let ${\cal CA}_i(t)$ be the set CAs who request content from CA $i$ at time $t$. 
For each $j \in {\cal CA}_i(t)$, CA $i$ first checks if it has any under-explored content $c$ for $p_j(t)$, i.e., $c$ such that $N^i_{c,p_j(t)}(t) \leq H_1(t)$. If so, it randomly selects one of its under-explored content to match it with the user of CA $j$. Otherwise, it exploits its content in ${\cal C}_i$ with the highest estimated relevance score for CA $j$'s current user's context, i.e.,
\begin{align}
b_{i,j}(t) \in \argmax_{c \in {\cal C}_i} \bar{r}^i_{c,p_j(t)}(t). \label{eqn:maximizer2}
\end{align}

A summary of notations used in the description of DISCOM is given in Table \ref{tab:notation2}.
The following theorem provides a bound on the regret of DISCOM. 

\newc{
\begin{table}
\centering
{
{\fontsize{9}{10}\selectfont
\begin{tabular}{|l|}
\hline
$L$: Similarity constant. $\gamma$: Similarity exponent \\
\hline
$T$: Time horizon \\
\hline
$m_T$: Slicing level of DISCOM \\
\hline
${\cal P}_T$: DISCOM's partition of ${\cal X}$ into $(m_T)^d$ hypercubes \\
\hline
$p_i(t)$: Hypercube in ${\cal P}_T$ that contains $x_i(t)$ \\
\hline
$N^i_p(t)$: Number of {\em all} user arrivals to CA $i$ with contexts in \\ $p \in {\cal P}_T$ by time $t$ except the training phases of CA $i$ \\
\hline
$N^i_{c,p}(t)$: Number of times content $c$ is selected in response to \\ a user arriving to CA $i$ with context in hypercube $p$ by time $t$ \\
\hline
$N^{\textrm{tr}, i}_{j,p}(t)$: Estimate of CA $i$ on the number of user arrivals with \\ contexts in $p$ to CA $j$ from all CAs except the training phases \\ of CA $j$ and exploration, exploitation phases of CA $i$ \\
\hline
$N^i_{j,p}(t)$: Number of users of CA $i$ with contexts in $p$ for which \\ content is requested from CA $j$ at exploration and exploitation \\ phases of CA $i$ by time $t$ \\
\hline
$H_1(t), H_2(t), H_3(t)$: Control functions of DISCOM \\
\hline
${\cal C}^{\textrm{ue}}_{i,p}(t)$: Set of under-explored content in ${\cal C}_i$ \\
\hline
${\cal M}^{\textrm{ct}}_{i,p}(t)$: Set of training candidates of CA $i$ \\
\hline
${\cal M}^{\textrm{ut}}_{i,p}(t)$: Set of CAs under-trained by CA $i$ \\
\hline
${\cal M}^{\textrm{ue}}_{i,p}(t)$: Set of CAs under-explored by CA $i$ \\
\hline
$\bar{r}^i_{k,p}(t)$: Sample man relevance score of action $k$ of CA $i$ at time $t$ \\
\hline
$\hat{\mu}^i_{k,p}(t)$ Estimated net reward of action $k$ of CA $i$ at time $t$ \\
\hline
\end{tabular}
}
}
\caption{Notations used in definition of DISCOM.}
\vspace{-0.3in}
\label{tab:notation2}
\end{table}
}

\begin{theorem}\label{theorem:cos}
When DISCOM is run by all CAs with parameters $H_1(t) = t^{2\gamma/(3\gamma+d)} \log t$, $H_2(t) = C_{\max} t^{2\gamma/(3\gamma+d)} \log t$, $H_3(t) = t^{2\gamma/(3\gamma+d)} \log t$ and $ m_T  = \left\lceil T^{1/(3\gamma + d)} \right\rceil$,\footnote{For a number $r \in \mathbb{R}$, let $\lceil r  \rceil$ be the smallest integer that is greater than or equal to $r$.} we have
\begin{align*}
& R_i(T) \leq  4 (M + C_{\max} +1 ) \beta_2 \\
&+ T^{\frac{2\gamma+d}{3\gamma+d}}
\left( \frac{ 14 L d^{\gamma/2}+12 + 4 (|{\cal C}_i| +M) M C_{\max} \beta_2    }{(2\gamma+d)/(3\gamma+d)} \right. \\
& \left. + 2^{d+1} Z_i \log T \right) \\
&+ T^{\frac{d}{3\gamma+d}} 2^{d+1} ( |{\cal C}_i| + 2 (M-1)  ),
\end{align*}
i.e., $R_i(T) = \tilde{O} \left(M C_{\max} T^{\frac{2\gamma+d}{3\gamma+d}} \right)$,\footnote{$\tilde{O}(\cdot)$ is the Big-O notation in which the terms with logarithmic growth rates are hidden.}
where $Z_i = |{\cal C}_i| + (M-1)(C_{\max}+1)$.
\end{theorem}
\begin{proof}
\jremove{The proof is given Appendix \ref{app:mainproof}.}
\aremove{The proof is given in our online technical report \cite{tekinTMM}.}
\end{proof}

For any $d>0$ and $\gamma>0$, the regret given in Theorem \ref{theorem:cos} is sublinear in time (or number of user arrivals). This guarantees that the regret per-user, i.e., the time averaged regret, converges to $0$
%\begin{align}
($\lim_{T \rightarrow \infty} \mathrm{E} [R_i(T)]/T = 0$). 
%\end{align}
It is also observed that the regret increases in the dimension $d$ of the context. 
By Assumption \ref{ass:lipschitz2}, a context is similar to another context if they are similar in each dimension, hence number of hypercubes in the partition ${\cal P}_T$ increases with $d$. 

In our analysis of the regret of DISCOM we assumed that $T$ is fixed and given as an input to the algorithm. DISCOM can be made to run independently of the final time $T$ by using a standard method called {\em the doubling trick} (see, e.g., \cite{slivkins2009contextual}). The idea is to divide time into rounds with geometrically increasing lengths and run a new instance of DISCOM at each round. For instance, consider rounds $\tau \in \{1,2,\ldots\}$, where each round has length $2^{\tau}$. Run a new instance of DISCOM at the beginning of each round with time parameter $2^\tau$. This modified version will also have $\tilde{O}\left(T^{(2\gamma+d)/(3\gamma+d)}\right)$ regret. 

\newc{Maximizing the satisfaction of an individual user is as important as maximizing the overall satisfaction of all users. 
The next corollary shows that by using DISCOM, CAs guarantee that their users will almost always be provided with the best content available within the entire content network.}
\newc{
\begin{corollary} \label{cor:confidence}
Assume that DISCOM is run with the set of parameters given in Theorem \ref{theorem:cos}. When DISCOM is in exploitation phase for CA $i$, we have
\begin{align}
& \mathrm{P} ( \mu^i_{a_i(t)}(x_i(t)) < \mu^i_{k^*_i(x_i(t))}(x_i(t))  - \delta_T ) \notag \\
& \leq \frac{2 |{\cal K}_i|}{t^2} + \frac{2 |{\cal K}_i| M C_{\max} \beta_2}{t^{\gamma/(3\gamma+d)}}. \notag
\end{align}
where $\delta_T = (6Ld^{\gamma/2} + 6) T^{-\gamma/(3\gamma+d)}$.
\end{corollary}
\begin{proof}
\jremove{The proof is given Appendix \ref{app:confidence}.}
\aremove{The proof is given in our online technical report \cite{tekinTMM}.}
\end{proof}
}
\newc{
\begin{remark}\label{rem:differentiable}
\textbf{(Differential Services)} 
Maximizing the performance for an individual user is particularly important for providing {\em differential services} based on the types of the users. For instance, a CA may want to provide higher quality recommendations to a subscriber (high type user) who has paid for the subscription compared to a non-subscriber (low type user). To do this, the CA can exploit the best content for the subscribed user, while perform exploration on a different user that is not subscribed. 
\end{remark}
}

\vspace{-0.1in}
\section{Regret When Feedback is Missing}\label{sec:missingfeedback}

When analyzing the performance of DISCOM, we assumed that the users always provide a feedback: like or dislike.
However, in most of the online content aggregation platforms user feedback is not always available. 
In this section we consider the effect of missing feedback on the performance of the proposed algorithm. 
We assume that each user gives a feedback with probability $p_r$ (which is unknown to the CAs). If the user at time $t$ does not give feedback, we assume that DISCOM does not update its counters. This will result in a larger number of trainings and explorations compared to the case when feedback is always available. 
The following theorem gives an upper bound on the regret of DISCOM for this case.

\begin{theorem} \label{thm:nolabel}
Let the DISCOM algorithm run with parameters $H_1(t) = t^{2\gamma/(3\gamma+d)} \log t$, $H_2(t) = C_{\max} t^{2\gamma/(3\gamma+d)} \log t$, $H_3(t) = t^{2\gamma/(3\gamma+d)} \log t$, and $ m_T  = \left\lceil T^{1/(3\gamma + d)} \right\rceil$. Then, if a user reveals its feedback with probability $p_r$, we have for CA $i$,
\begin{align*}
& R_i(T) \leq  4 (M + C_{\max} +1 ) \beta_2 \\
&+ T^{\frac{2\gamma+d}{3\gamma+d}}
\left( \frac{ 14 L d^{\gamma/2}+12 + 4 (|{\cal C}_i| +M) M C_{\max} \beta_2    }{(2\gamma+d)/(3\gamma+d)} \right. \\
& \left. + \frac{2^{d+1} Z_i}{p_r} \log T \right) \\
&+ T^{\frac{d}{3\gamma+d}} 2^{d+1} \frac{|{\cal C}_i| + 2 (M-1)} {p_r}  ,
\end{align*}
i.e., 
%\begin{align}
$R_i(T) = \tilde{O} \left(M C_{\max} T^{\frac{2\gamma+d}{3\gamma+d}}/p_r \right)$  ,    \notag
%\end{align}
where $Z_i = |{\cal C}_i| + (M-1)(C_{\max}+1)$, $\beta_a := \sum_{t=1}^{\infty} 1/t^a$.
\end{theorem}
\begin{proof}
\jremove{The proof is given Appendix \ref{app:theorem2}.}
\aremove{The proof is given in our online technical report \cite{tekinTMM}.}
\end{proof}

From Theorem \ref{thm:nolabel}, we see that missing feedback does not change the time order of the regret. However, the regret is scaled with $1/p_r$, which is the expected number of users required for a single feedback.

\vspace{-0.1in}
\section{Learning Under Dynamic User and Content Characteristics} \label{sec:dynamic}

When the user and content characteristics change over time, the relevance score of content $c$ for a user with context $x$ changes over time. 
In this section, we assume that the following relation holds between the probabilities that a content will be liked with users with similar contexts at two different times $t$ and $t'$. 
\begin{assumption} \label{ass:lipschitz3}
For each $c \in {\cal C}$, there exists $L>0$, $\gamma>0$ such that for all $x,x' \in {\cal X}$, we have
\begin{align}
|\pi_{c,t}(x) - \pi_{c,t'}(x')| \leq L \left(   ||x-x'|| \right)^\gamma  + |t - t'|/T_s    \notag
\end{align}
where $1/T_s > 0$ is the speed of the change in user and content characteristics. We call $T_s$ the {\em stability parameter}.
\end{assumption}

Assumption \ref{ass:lipschitz3} captures the temporal dynamics of content matching which is absent in Assumption \ref{ass:lipschitz2}.
Such temporal variations are often referred to as {\em concept drift} \cite{gao2007appropriate, masud2009integrating}. 
When there is concept drift, a learner should also consider which past information to take into account when learning, in addition to how to combine the past information to learn the best matching strategy. 

The following modification of DISCOM will deal with dynamically changing user and content characteristics by using a time window of past observations in estimating the relevance scores.  
The modified algorithm is called DISCOM with time window (DISCOM-W). This algorithm groups the time slots into rounds $\zeta=1,2,\ldots$ each having a fixed length of $2 \tau_h$ time slots, where $\tau_h$ is an integer called {\em the half window length}.
Some of the time slots in these rounds overlap with each other as given in Fig. \ref{fig:ACAPW}.
 The idea is to keep separate control functions and counters for each round, and calculate the sample mean relevance scores for groups of similar contexts based only on the observations that are made during the time window of that round. 
We call $\eta=1$ the {\em initialization round}. 
The control functions for the initialization round of DISCOM-W is the same as the control functions $H_1(t)$, $H_2(t)$ and $H_3(t)$ of DISCOM whose values are given in Theorem \ref{theorem:cos}.
For the other rounds $\zeta>1$, the control functions depend on $\tau_h$ and are given as
\begin{align}
H^{\tau_h}_1(t) = H^{\tau_h}_3(t) = (t \mod \tau_h +1)^z \log (t \mod \tau_h +1)  \notag
\end{align}
and
\begin{align}
H^{\tau_h}_2(t) = C_{\max} (t \mod \tau_h +1)^z \log (t \mod \tau_h +1) \notag 
\end{align}
for some $0<z<1$. 
Each round $\eta$ is divided into two sub-rounds. Except the initialization round, i.e., $\eta=1$, the first sub-round is called the {\em passive} sub-round, while the second sub-round is called the {\em active} sub-round. For the initialization round both sub-rounds are active sub-rounds. 
In order to reduce the number of trainings and explorations, DISCOM-W has an overlapping round structure as shown in Fig. \ref{fig:ACAPW}. For each round except the initialization round, passive sub-rounds of round $\eta$, overlaps with the active sub-round of round $\eta-1$.
DISCOM-W operates in the same way as DISCOM in each round.
DISCOM-W can be viewed as an algorithm which generates a new instance of DISCOM at the beginning of each round, with the modified control functions.
DISCOM-W runs two different instances of DISCOM at each round.
One of these instances is the active instance based on which content matchings are performed, and the other one is the passive instance which learns through the content matchings made by the active instance. 

 Let the instance of DISCOM that is run by DISCOM-W at round $\eta$ be $\textrm{DISCOM}_{\eta}$. 
The hypercubes of $\textrm{DISCOM}_{\eta}$ are formed in a way similar to DISCOM's.
The input time horizon is taken as $T_s$ which is the stability parameter given in Assumption \ref{ass:lipschitz3}, and the slicing parameter $m_{T_s}$ is set accordingly. $\textrm{DISCOM}_{\eta}$ uses the partition of ${\cal X}$ into $(m_{T_s})^d$ hypercubes denoted by ${\cal P}_{T_s}$.
When all CAs are using DISCOM-W, the matching action selection of CA $i$ only depends on the history of content matchings and feedback observations at round $\eta$. 
If time $t$ is in the active sub-round of round $\eta$, matching action of CA $i \in {\cal M}$ is taken according to $\textrm{DISCOM}_{\eta}$. As a result of the content matching, sample mean relevance scores and counters of both $\textrm{DISCOM}_{\eta}$ and $\textrm{DISCOM}_{\eta+1}$ are updated. Else if time $t$ is in the passive sub-round of round $\eta$, matching action of CA $i \in {\cal M}$ is taken according to $\textrm{DISCOM}_{\eta-1}$ (see Fig. \ref{fig:ACAPW}). As a result of this, sample mean relevance scores and counters of both $\textrm{DISCOM}_{\eta-1}$ and $\textrm{DISCOM}_{\eta}$ are updated.

\begin{figure}
\begin{center}
\includegraphics[width=0.95\columnwidth]{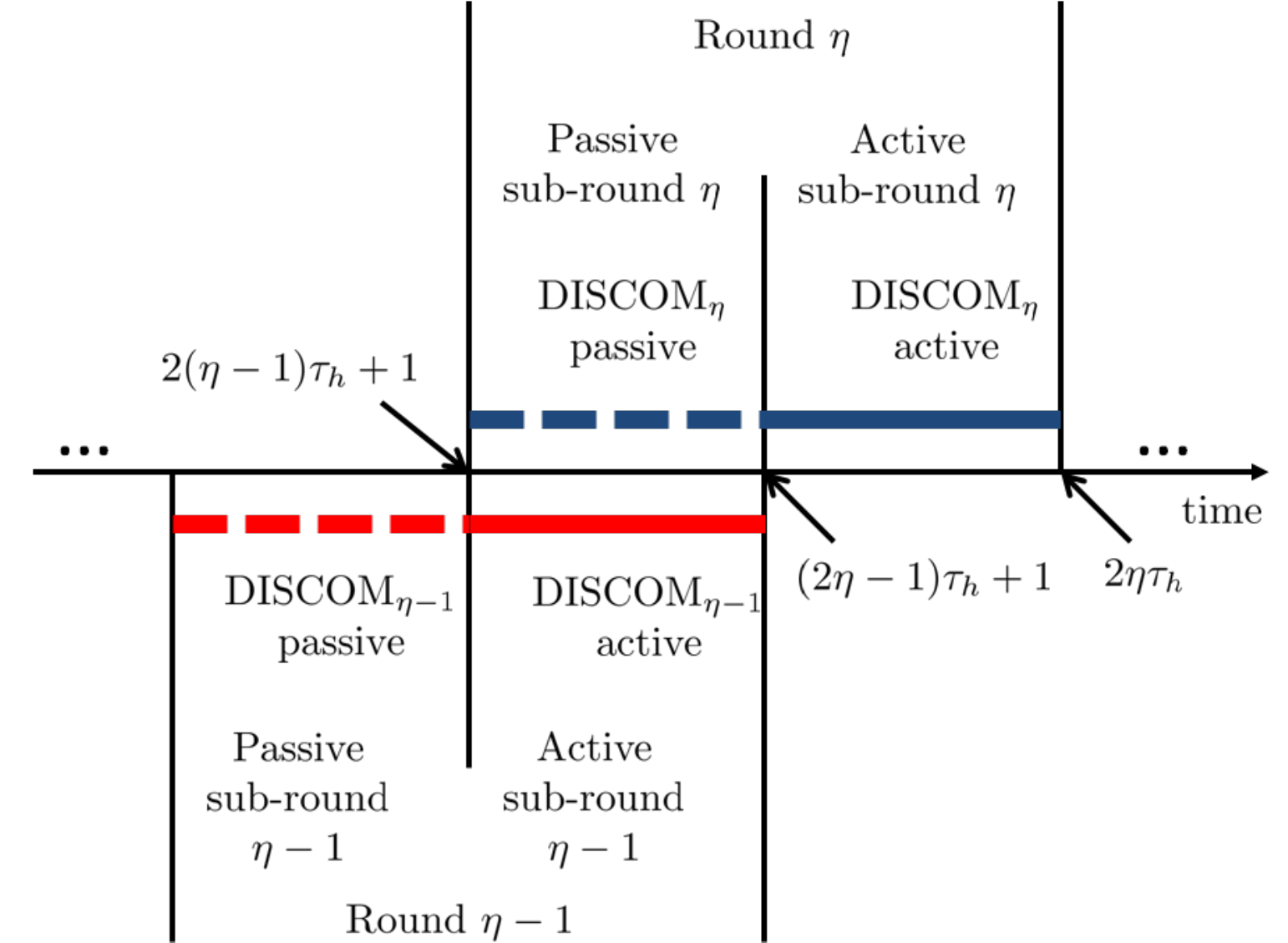}
\vspace{-0.1in}
\caption{Operation of DISCOM-W showing the round structure and the different instances of DISCOM running for each round.} 
\vspace{-0.2in}
\label{fig:ACAPW}
\end{center}
\end{figure}

At the start of a round $\eta$, the relevance score estimates and counters for $\textrm{DISCOM}_{\eta}$ are equal to zero. However, due to the two sub-round structure, 
when the active sub-round of round $\eta$ starts, CA $i$ already has some observations for the context and actions taken in the passive sub-round of that round, hence depending on the arrivals and actions in the passive sub-round, the CA may even start the active sub-round by exploiting, whereas it should have always spent some time in training and exploration if it starts an active sub-round without any past observations (cold start problem).

In this section, due to the concept drift, even though the context of a past user can be similar to the context of the current user, their relevance scores for a content $c$ can be very different. Hence DISCOM-W assumes that a past user is similar to the current user only if it arrived in the current round. Since round length is fixed, it is impossible to have sublinear number of similar context observations for every $t$. 
Thus, achieving sublinear regret under concept drift is not possible. Therefore, in this section we focus on the average regret which is given by 
\begin{align}
& R^{\textrm{avg}}_i(T) := \frac{1}{T} \sum_{t=1}^T \left( \pi_{k_i^*(x_i(t))}(x_i(t)) - d^i_{k_i^*(\boldsymbol{x}_i(t))} \right)  \notag \\
&- \frac{1}{T} \mathrm{E} \left[ \sum_{t=1}^T \left( \mathrm{I} (   y_i(t) = L ) - d^i_{a_i(t)} \right) \right].     \notag
\end{align}

The following theorem bounds the average regret of DISCOM-W.
\begin{theorem}\label{thm:adaptivemain2}
When DISCOM-W is run with parameters
\begin{align}
H^{\tau_h}_1(t) &= H^{\tau_h}_3(t) = (t\textrm{ mod }\tau_h +1)^{\frac{2\gamma}{3\gamma+d}} \log ( t\textrm{ mod }\tau_h +1)     \notag \\
H^{\tau_h}_2(t) &= C_{\max} (t\textrm{ mod }\tau_h +1)^{\frac{2\gamma}{3\gamma+d}} \log (t\textrm{ mod }\tau_h +1)       \notag
\end{align}
$m_{T_s} = \lceil T_s^{\frac{1}{3\gamma+d}} \rceil$ and $\tau_h = \lfloor T_s^{(3\gamma+d)/(4\gamma+d)} \rfloor$,\footnote{For a number $b$, $\lfloor b \rfloor$ denotes the largest integer that is smaller than or equal to $b$.} where $T_s$ is the stability parameter which is given in Assumption \ref{ass:lipschitz3}, the time averaged regret of CA $i$ by time $T$ is 
\begin{align*}
R^{\textrm{avg}}_i(T) = \tilde{O} \left(     T_s^{ \frac{-\gamma}{4\gamma+d}} \right)
\end{align*}
for any $T >0$. Hence DISCOM-W is $\epsilon = \tilde{O} \left(     T_s^{ \frac{-\gamma}{4\gamma+d}} \right)$ approximately optimal in terms of the average reward.
%
%\begin{align*}
%&R_i(T) \leq T^{f_1(\alpha)} \left( 8 D Z_i \log T + \frac{2 B L A 2^{2+\alpha+\sqrt{9\alpha^2 + 8\alpha}}}{2^{\frac{2+\alpha+\sqrt{9\alpha^2 + 8\alpha}}{2}}-1} \right) \\
%&+  T^{f_2(\alpha)} 8(M-1)C_{\max} \beta_2 (3\alpha + \sqrt{9\alpha^2+8\alpha})/(4\alpha) \\
%& + T^{f_3(\alpha)} (8 D Z_i 
%+ 4 (M-1) C_{\max} \beta_2) + 2(1+D) |{\cal F}_i| \beta_2,
%\end{align*}
%
%where $Z_i = |{\cal F}_i| + (M-1) (C_{\max}+1)$, 
%$f_1(\alpha) = (2+\alpha+\sqrt{9\alpha^2+8\alpha})/(2+3\alpha+\sqrt{9\alpha^2+8\alpha})$,
%$f_2(\alpha) = 2\alpha/(3\alpha+\sqrt{9\alpha^2+8\alpha})$,
%$f_3(\alpha) = 2/(2+3\alpha+\sqrt{9\alpha^2+8\alpha})$.
%
\end{theorem}
\begin{proof}
\jremove{The proof is given Appendix \ref{app:theorem3}.}
\aremove{The proof is given in our online technical report \cite{tekinTMM}.}
\end{proof}

From the result of this theorem we see that the average regret decays as the stability parameter $T_s$ increases. This is because, DISCOM-W will use a longer time window (round) when $T_s$ is large, and thus can get more observations to estimate the sample mean relevance scores of the matching actions in that round, which will result in better estimates hence smaller number of suboptimal matching action selections. Moreover, the average number of trainings and explorations required decrease with the round length.

\section{Numerical Results} \label{sec:numerical}

In this section we provide numerical results for our proposed algorithms DISCOM and DISCOM-W on real-world datasets.

\subsection{Datasets}

\newc{For all the datasets below, for a CA the cost of choosing a content within the content network and the cost of choosing another CA is set to $0$. Hence, the only factor that affects the total reward is the users' ratings for the contents.}

\newc{
\textbf{Yahoo! Today Module (YTM)} \cite{li2010contextual}: The dataset contains news article webpage recommendations of Yahoo! Front Page. Each instance is composed of (i) ID of the recommended content, (ii) the user's context ($2$-dimensional vector), (iii) the user's click information.  
The user's click information for a webpage/content is associated with the relevance score of that content. It is equal to $1$ if the user clicked on the recommended webpage and $0$ else. The dataset contains $T=70000$ instances and $40$ different types of content. 
We generate $4$ CAs and assign $10$ of the $40$ types of content to each CA's content network. Each CA has direct access to content in its own network, while it can also access to the content in other CAs' content network by requesting content from these CAs. 
Users are divided into four groups according to their contexts and each group is randomly assigned to one of the CAs. Hence, the user arrival processes to different CA's are different.
The performance of a CA is evaluated in terms of the average number of clicks, i.e., click through rate (CTR), of the contents that are matched with its users.
}

\newc{
\textbf{Music Dataset (MD)}: The dataset contains contextual information and ratings (like/dislike) of music genres (classical, rock, pop, rap) collected from 413 students at UCLA.
We generate $2$ CAs each specialized in two of the four music genres. Users among the 413 users randomly arrive to each CA. A CA either recommends a music content that is in its content network or asks another CA, specialized in another music genre, to provide a music item. As a result, the rating of the user for the genre of the provided music content is revealed to the CA. 
The performance of a CA is evaluated in terms of the average number of likes it gets for the contents that are matched with its users.
}

\newc{
\textbf{Yahoo! Today Module (YTM) with Drift (YTMD)}: This dataset is generated from YTM to simulate the scenario where the user ratings for a particular content changes over time. After every $10000$ instances, $20$ contents are randomly selected and user clicks for these contents are set to $0$ (no click) for the next $10000$ instances. For instance, this can represent a scenario where some news articles lose their popularity a day after they become available while some other news articles related to ongoing events will stay popular for several days. 
}
\vspace{-0.1in}
\subsection{Learning Algorithms}\label{sec:numalgo}

\newc{
While DISCOM and DISCOM-W are the first distributed algorithms to perform content aggregation (see Table \ref{tab:comparison}), we compare their performance with distributed versions of the centralized algorithms proposed in \cite{li2010contextual,chu2011contextual,bouneffouf2012hybrid,slivkins2009contextual}.
In the distributed implementation of these centralized algorithms, we assume that each CA runs an independent instance of these algorithms. For instance, when implementing a centralized algorithm $A$ on the distributed system of CAs, we assume that each CA $i$ runs its own instance of $A$ denoted by $A_i$.
When CA $i$ selects CA $j$ as a matching action in ${\cal K}_i$ by using its algorithm $A_i$, CA $j$ will select the content for CA $i$ using its algorithm $A_j$ with CA $i$'s user's context on the set of contents ${\cal C}_j$. In our numerical results, each algorithm is run for different values of its input parameters. The results are shown for the parameter values for which the corresponding algorithm performs the best. 
}

\newc{
\textbf{DISCOM}: Our algorithm given in Fig. \ref{fig:CLUP} with control functions $H_1(t)$, $H_2(t)$ and $H_3(t)$ divided by $10$ for MD, and by $20$ for YTM and YTMD to reduce the number of trainings and explorations.\footnote{The number of trainings and explorations required in the regret bounds are the worst-case numbers. In reality, good performance is achieved with a much smaller number of trainings and explorations.} 
}

\newc{
\textbf{DISCOM-W}: Our algorithm given in Fig. \ref{fig:ACAPW} which is the time-windowed version of DISCOM with control functions $H_1(t)$, $H_2(t)$ and $H_3(t)$ divided by $20$ to reduce the number of trainings and explorations.
}

\newc{As we mentioned in Remark \ref{rem:differentiable}, both DISCOM and DISCOM-W can provide differential services to its users. In this case both algorithms always exploit for the users with high type (subscribers) and if necessary can train and explore for the users with low type (non-subscribers). Hence, the performance of DISCOM and DISCOM-W for differential services is equal to their performance for the set of high type users. 
}

\newc{
\textbf{LinUCB} \cite{li2010contextual,chu2011contextual}: This algorithm computes an index for each matching action by assuming that the relevance score of a matching action for a user is a linear combination of the contexts of the user. Then for each user it selects the matching action with the highest index.
}

\newc{
\textbf{Hybrid-$\epsilon$} \cite{bouneffouf2012hybrid}:
This algorithm forms context-dependent sample mean rewards for the matching actions by considering the history of observations and decisions for groups of contexts that are similar to each other. For user $t$ it either explores a random matching action with probability $\epsilon_t$ or exploits the best matching action with probability $1-\epsilon_t$, where $\epsilon_t$ is decreasing in $t$. 
}

\newc{
\textbf{Contextual zooming (CZ)} \cite{slivkins2009contextual}: This algorithm adaptively creates balls over the joint action and context space, calculates an index for each ball based on the history of selections of that ball, and at each time step selects a matching action according to the ball with the highest index that contains the current context. 
}

\vspace{-0.1in}
\subsection{Yahoo! Today Module Simulations}

\newc{
In YTM each instance (user) has two contexts $(x_1, x_2) \in [0,1]^2$. We simulate the algorithms in Section \ref{sec:numalgo} for three different context sets in which the learning algorithms only decide based on (i) the first context $x_1$, (ii) the second context $x_2$, and (iii) both contexts $(x_1, x_2)$ of the users. 
The $m_T$ parameter of DISCOM for these simulations is set to the optimal value found in Theorem 1 (for $\gamma=1$) which is $\lceil T^{1/4} \rceil$ for simulations with a single context and $\lceil T^{1/5} \rceil$ for simulations with both contexts. DISCOM is run for numerous $z$ values ranging from $1/4$ to $1/2$. 
Table \ref{tab:yahoosim} compares the performance of DISCOM, LinUCB, Hybrid-$\epsilon$ and CZ. All of the algorithms are evaluated at the parameter values in which they perform the best. 
As seen from the table the CTR for DISCOM with differential services is $16\%$, $5\%$ and $7\%$  higher than the best of LinUCB, Hybrid-$\epsilon$ and CZ for contexts $x_1$, $x_2$ and $(x_1, x_2)$, respectively. 
%As we mentioned in Section \ref{sec:iid}, the exploitation performance is important for applications in which the CAs are required to provide differential services to ensure the satisfaction of higher type users (subscribers). 
%As expected DISCOM-W performs the best in this dataset in terms of the average number of clicks for the users it exploits, with about $23\%$, $11.3\%$ and $51.6\%$ improvement over the best of LinUCB, Hybrid-$\epsilon$ and CZ, for types of contexts $x_1$, $x_2$, and $(x_1, x_2)$, respectively.
}

\begin{table}
\centering
{
{\fontsize{9}{9}\selectfont
\setlength{\tabcolsep}{.1em}
\begin{tabular}{|c|c|c|c|c|c|}
\hline
Context & DISCOM & DISCOM & LinUCB  & Hybrid-$\epsilon$ & CZ  \\
  & & (diff. serv.) & & & \\
\hline
$x_1$ &  6.37  &  7.30  & 6.31 &  5.92 & 4.29 \\
$x_2$ &  6.14   &  6.45  & 4.72  & 6.14 & 4.39 \\
$(x_1, x_2)$ & 5.93  & 6.61  &  5.65 & 6.15  & 4.24 \\
\hline
\end{tabular}
}
}
\caption{Comparison of the CTR$\times 10^{2}$  achieved by CA 1 for DISCOM and other learning algorithms for YTM.}
\vspace{-0.3in}
\label{tab:yahoosim}
\end{table}

\newc{
Table \ref{tab:musicsweep} compares the performance of DISCOM, the percentage of training, exploration and exploitation phases for different control functions (different $z$ parameters) for simulations with context $x_2$. As expected, the percentage of trainings and explorations increase with the control function. As $z$ increases matching actions are explored with a higher accuracy, and hence the average exploitation reward (CTR) increases. 
}

\begin{table}
\centering
{
{\fontsize{9}{9}\selectfont
\setlength{\tabcolsep}{.1em}
\begin{tabular}{|c|c|c|c|}
\hline
 $z$ & 1/4 & 1/3 & 1/2  \\
 \hline
CTR$\times10^{2}$ & 5.13 & 5.29 & 6.14 \\
CTR$\times10^{2} $ in exploitations  & 5.14 & 5.34 & 6.45 \\
Exploit \% & 98.9 & 97.7 & 90.2 \\
Explore \% & 0.5 & 0.6 & 1.9 \\
Train \% & 0.7 & 1.7 & 7.9 \\
\hline
\end{tabular}
}
}
\caption{The CTR, training, exploration and exploitation percentages of CA 1 using DISCOM with context $x_2$ for YTM.}
\vspace{-0.3in}
\label{tab:musicsweep}
\end{table}

\vspace{-0.1in}
\subsection{Music Dataset Simulations}

\newc{
Table \ref{tab:musicsim} compares the performance of DISCOM, LinUCB, Hybrid-$\epsilon$ and CZ for the music dataset. The parameter values used for DISCOM for the result in Table \ref{tab:musicsim} are $z=1/8$ and $m_T= 4$. From the results is is observed that DISCOM achieves $10\%$ improvement over LinUCB, $5\%$ improvement over Hybrid-$\epsilon$, and $28\%$ improvement over CZ in terms of the average number of likes achieved for the users of  CA $1$. Moreover, the average number of likes received by DISCOM for the high type users (differential services) is even higher, which is $13\%$, $8\%$ and $32\%$ higher than LinUCB, HE and CZ, respectively.  
}

\begin{table}
\centering
{
{\fontsize{9}{9}\selectfont
\setlength{\tabcolsep}{.1em}
\begin{tabular}{|c|c|c|c|c|c|}
\hline
Algorithm & DISCOM & DISCOM & LinUCB  & Hybrid-$\epsilon$ & CZ  \\
  & & (diff. serv.) & & & \\
\hline
Avg. num. & 0.717  &  0.736  & 0.652 & 0.683 & 0.559 \\
of likes               &  &&  & & \\
\hline
\end{tabular}
}
}
\caption{Comparison among DISCOM and other learning algorithms for MD.}
\vspace{-0.3in}
\label{tab:musicsim}
\end{table}

\vspace{-0.1in}
\subsection{Yahoo! Today Module with Drift Simulations}

\newc{
Table \ref{tab:driftsim} compares the performance of DISCOM-W with half window length ($\tau_h=2500$) and $m_T=10$, DISCOM (with $m_T$ set equal to $\lceil T^{1/4} \rceil$ simulations with a single context dimension and $\lceil T^{1/5} \rceil$ for the simulation with two context dimensions),  LinUCB, Hybrid-$\epsilon$ and CZ. For the results in the table, the $z$ parameter value of DISCOM and DISCOM-W are set to the $z$ value in which they achieve the highest number of clicks. Similarly, LinUCB, Hybrid-$\epsilon$ and CZ are also evaluated at their best parameter values. 
The results show the performance of DISCOM and DISCOM-W for differential services. DISCOM-W performs the best in this dataset in terms of the average number of clicks, with about $23\%$, $11.3\%$ and $51.6\%$ improvement over the best of LinUCB, Hybrid-$\epsilon$ and CZ, for types of contexts $x_1$, $x_2$, and $(x_1, x_2)$, respectively.
}

\begin{table}
\centering
{
{\fontsize{9}{9}\selectfont
\setlength{\tabcolsep}{.1em}
\begin{tabular}{|c|c|c|c|c|c|}
\hline
Contexts  & DISCOM-W & DISCOM & LinUCB  & Hybrid-$\epsilon$ & CZ  \\
used  & (diff. serv.) & (diff. serv.) & & & \\
\hline
 $x_1$   &  6.3 & 5.5  & 5.1 & 3.0  & 2.4 \\
 $x_2$    &  5.1  & 4.2  & 4.2  & 4.6  & 2.4 \\
 $(x_1,x_2)$    &  6.9  & 3.8  & 4.6 & 4.1 & 2.3 \\
\hline
\end{tabular}
}
}
\caption{The CTR$\times 10^{2}$ of DISCOM-W and DISCOM for differential services, and the CTR of other learning algorithms for YTMD.}
\vspace{-0.3in}
\label{tab:driftsim}
\end{table}

% no \IEEEPARstart
\vspace{-0.1in}
\section{Conclusion} \label{sec:conc}
In this paper we considered novel online learning algorithms for content matching by a distributed set of CAs. We have characterized the relation between the user and content characteristics in terms of a relevance score, and proposed online learning algorithms that learns to match each user with the content with the highest relevance score.
When the user and content characteristics are static, the best matching between content and each type of user can be learned perfectly, i.e., the average regret due to suboptimal matching goes to zero. 
When the user and content characteristics are dynamic, depending on the rate of the change, an approximately optimal matching between content and each user type can be learned.
In addition to our theoretical results, we have validated the concept of distributed content matching on real-world datasets.
An interesting future research direction is to investigate the interaction between different CAs when they compete for the same pool of users. Should a CA send a content that has a high chance of being liked by another CA's user to increase its immediate reward, or should it send a content that has a high chance of being disliked by the other CA's user to divert that user from using that CA and switch to it instead. 

\jremove{
\appendices
\vspace{-0.15in}
 \section{A bound on divergent series} \label{app:seriesbound}
 For $p>0$, $p \neq 1$,
\begin{align*}
\sum_{t=1}^{T} t^{-p} \leq 1 + (T^{1-p} -1) / (1-p).
\end{align*}
\begin{proof}
See \cite{chlebus2009approximate}.
\end{proof}
\vspace{-0.15in}
 \section{Proof of Theorem \ref{theorem:cos}} \label{app:mainproof}
 \subsection{Necessary Definitions and Notations}

Let $\beta_a := \sum_{t=1}^{\infty} 1/t^a$, and let $\log(.)$ denote logarithm in base $e$.
For each hypercube $p \in {\cal P}_T$ let 
%
%\begin{align*}
$\overline{\pi}_{c,p} := \sup_{x \in p} \pi_c(x)$,
$\underline{\pi}_{c,p} := \inf_{x \in p} \pi_c(x)$, for $c \in {\cal C}$, and
$\overline{\mu}^i_{k,p} := \sup_{x \in p} \mu^i_k(x)$,
$\underline{\mu}^i_{k,p} := \inf_{x \in p} \mu^i_k(x)$, for $k \in {\cal K}_i$.
%\end{align*}
%
Let $x^*_p$ be the context at the center (center of symmetry) of the hypercube $p$. We define the optimal matching action of CA $i$ for hypercube $p$ as
%
%\begin{align*}
$k^*_i(p) := \argmax_{k \in {\cal K}_i} \mu^i_k(x^*_p)$.
%\end{align*}
%
When the hypercube $p$ is clear from the context, we will simply denote the optimal matching action for hypercube $p$ with $k^*_i$.
Let
\begin{align*}
{\cal L}^i_p(t) := \left\{ k \in {\cal K}_i : \underline{\mu}^i_{k^*_i(p),p} - \overline{\mu}^i_{k,p} > (4 L d^{\gamma/2} + 6) t^{-z/2}  \right\}
\end{align*}
be the set of suboptimal matching actions of CA $i$ at time $t$ in hypercube $p$.
Also related to this let
\begin{align*}
{\cal C}^j_p(t) := \left\{ c \in {\cal C}_j : \underline{\pi}_{c^*_j(p),p} - \overline{\pi}_{c,p} > (4 L d^{\gamma/2} + 6)  t^{-z/2} \right\}
\end{align*}
be the set of suboptimal contents of CA $j$ at time $t$ in hypercube $p$, where $c^*_j(p) = \argmax_{c \in {\cal C}_j} \pi_c(x^*_p)$. Also when the hypercube $p$ is clear from the context we will just use $c^*_j$.
The contents in ${\cal C}^j_p(t) $ are the ones that CA $j$ should not select when called by another CA.
The regret given in (3) can be written as a sum of three components: 
\begin{align}
R_i(T) = \mathrm{E} [R^e_i(T)] + \mathrm{E} [R^s_i(T)] + \mathrm{E} [R^n_i(T)]      \notag
\end{align}
where $R^e_i(T)$ is the regret due to trainings and explorations by time $T$, $R^s_i(T)$ is the regret due to suboptimal matching action selections in exploitations by time $T$ and $R^n_i(T)$ is the regret due to near optimal matching action selections in exploitations by time $T$, which are all random variables. 

\vspace{-0.1in}
\subsection{Bounding the Regret in Training, Exploration and Exploitation phases.}
In the following lemmas we will bound each of these terms separately. The following lemma bounds $\mathrm{E} [R^e_i(T)]$. 
%
%\rev{Due to space constraints some of the proofs are not included in the paper. For the complete proofs please see the online appendix \cite{}}.
%
\begin{lemma} \label{lemma:explorations}
Consider all CAs running DISCOM with parameters $H_1(t) = t^{z} \log t$, $H_2(t) = C_{\max} t^{z} \log t$, $H_3(t) = t^{z} \log t$ and $m_T = \left\lceil T^{\kappa} \right\rceil$, where $0<z<1$ and $0<\kappa<1/d$. Then, we have
\begin{align*}
\mathrm{E} [R^e_i(T)] &\leq  2^{d+1} (|{\cal C}_i| + (M-1) (C_{\max} + 1) ) T^{z + \kappa d} \log T \\
&+ 2^{d+1} ( |{\cal C}_i| + 2 (M-1) ) T^{\kappa d}  .
\end{align*}
\end{lemma}
\begin{proof}
Since time slot $t$ is a training or an exploration slot for CA $i$ if and only if 
\begin{align}
{\cal M}^{\textrm{ut}}_{i,p_i(t)}(t) \cup {\cal M}^{\textrm{ue}}_{i,p_i(t)}(t) \cup {\cal C}^{\textrm{ue}}_{i,p_i(t)}(t)  \neq \emptyset       \notag
\end{align}
up to time $T$, there can be at most   $\left\lceil T^{z} \log T \right\rceil$ exploration slots in which a content $c \in {\cal C}_i$ is matched with the user of CA $i$, 
$\left\lceil C_{\max} T^{z} \log T \right\rceil$ training slots in which CA $i$ selects CA $j \in {\cal M}_{-i}$, $\left\lceil T^{z} \log T \right\rceil$ exploration slots in which CA $i$ selects CA $j \in {\cal M}_{-i}$. Result follows from summing these terms and the fact that $(m_T)^d \leq 2^d T^{\kappa d}$ for any $T \geq 1$.
The additional factor of 2 comes from the fact that the realized regret at any time slot can be at most $2$.
\end{proof}

%From Lemma \ref{lemma:explorations}, we see that the regret due to trainings and explorations is linear in the number of hypercubes $(m_T)^d$, hence exponential in parameter $\kappa$ and $z$. We conclude that $z$ and $\kappa$ should be small enough to achieve sublinear regret in exploration slots.

For any $k \in {\cal K}_i$ and $p \in {\cal P}_T$, the sample mean $\bar{r}^i_{k,p}(t)$ of the relevance score of matching action $k$ represents a random variable which is the average of the independent samples in set ${\cal E}^i_{k,p}(t)$.
Since these samples are not identically distributed, in order to facilitate our analysis of the regret, we generate two different artificial i.i.d. processes to bound the probabilities related to  $\hat{\mu}^i_{k,p}(t) = \bar{r}^i_{k,p}(t) - d^i_k$, $k \in {\cal K}_i$.
The first one is the {\em best} process for CA $i$ in which the net reward of the matching action $k$ for a user with context in $p$ is sampled from an i.i.d. Bernoulli process with mean $\overline{\mu}^i_{k,p}$, the other one is the {\em worst} process for CA $i$ in which this net reward is sampled from an i.i.d. Bernoulli process with mean $\underline{\mu}^i_{k,p}$. Let $\hat{\mu}^{\textrm{b},i}_{k,p}(z)$ denote the sample mean of the $z$ samples from the best process and $\hat{\mu}^{\textrm{w},i}_{k,p}(z)$ denote the sample mean of the $z$ samples from the worst process for CA $i$.
We will bound the terms $\mathrm{E} [R^n_i(T)]$ and $\mathrm{E} [R^s_i(T)]$ by using these artificial processes along with the similarity information given in Assumption 1.

Let $\Xi^i_{j,p}(t)$ be the event that a suboptimal content $c \in {\cal C}_{j}$ is selected by CA $j \in {\cal M}_{-i}$, when it is called by CA $i$ for a context in set $p$ for the $t$th time in the exploitation phases of CA $i$.
Let $X^i_{j,p}(t)$ denote the random variable which is the number of times CA $j$ selects a suboptimal content when called by CA $i$ in exploitation slots of CA $i$ when the context is in set $p \in {\cal P}_T$ by time $t$.
Clearly, we have
\begin{align}
X^i_{j,p}(t) = \sum_{t'=1}^{|{\cal E}^i_{j,p}(t)|} \mathrm{I}(\Xi^i_{j,p}(t')) \notag
\end{align}
where $\mathrm{I}(\cdot)$ is the indicator function which is equal to $1$ if the event inside is true and $0$ otherwise.
The following lemma bounds $\mathrm{E} [R^s_i(T)]$.
\begin{lemma} \label{lemma:suboptimal1}
Consider all CAs running DISCOM with parameters $H_1(t) = t^{z} \log t$, $H_2(t) = C_{\max} t^{z} \log t$, $H_3(t) = t^{z} \log t$ and $m_T = \left\lceil T^{\kappa} \right\rceil$, where $0<z<1$ and $\kappa = z/(2\gamma)$. Then, we have
%\vspace{-0.1in}
%
\begin{align*}
\mathrm{E} [R^e_i(T)] &\leq  4 (|{\cal C}_i| + M) \beta_2  \\
&+ 4 (|{\cal C}_i| + M)  M C_{\max} \beta_2 \frac{T^{1-z/2}} {1-z/2} .
\end{align*}
\end{lemma}

\begin{proof} 
Consider time $t$. For simplicity of notation let $p = p_i(t)$. Let 
\begin{align*}
{\cal W}^i(t) := \{ {\cal M}^{\textrm{ut}}_{i,p_i(t)}(t) \cup {\cal M}^{\textrm{ue}}_{i,p_i(t)}(t) \cup {\cal C}^{\textrm{ue}}_{i,p_i(t)}(t)  = \emptyset  \}
\end{align*}
be the event that CA $i$ exploits at time $t$.

First, we will bound the probability that CA $i$ selects a suboptimal matching action in an exploitation slot.
Then, using this we will bound the expected number of times a suboptimal matching action is selected by CA $i$ in exploitation slots.
Note that every time a suboptimal matching action is selected by CA $i$, since $\mu^i_k(x) = \pi^i_k(x) - d^i_k \in [-1,1]$ for all $k \in {\cal K}_i$, the realized (hence expected) loss is bounded above by $2$.
Therefore $2$ times the expected number of times a suboptimal matching action is chosen in an exploitation slot bounds the regret due to suboptimal matching actions in exploitation slots.
Let ${\cal V}^i_{k}(t)$ be the event that matching action $k \in {\cal K}_i$ is chosen at time $t$ by CA $i$.
We have
\begin{align*}
R^s_i(T) \leq 2 \sum_{t=1}^T \sum_{k \in {\cal L}^i_{p_i(t)}(t)} \mathrm{I}(  {\cal V}^i_{k}(t), {\cal W}^i(t) ) .
\end{align*} 
Taking the expectation
\begin{align}
\mathrm{E}[R^s_i(T)] \leq 2 \sum_{t=1}^T \sum_{k \in {\cal L}^i_{p_i(t)}(t)} \mathrm{P}({\cal V}^i_{k}(t), {\cal W}^i(t) ). \label{eqn:subregret}
\end{align}
Let ${\cal B}^i_{j,p}(t)$ be the event that at most $t^{\phi}$ samples in ${\cal E}^i_{j,p}(t)$ are collected from suboptimal content of CA $j$.
Let ${\cal B}^i(t) := \bigcap_{j \in {\cal M}_{-i} } {\cal B}^i_{j,p_i(t)}(t)$.
For a set ${\cal A}$, let ${\cal A}^c$ denote the complement of that set.
For any $k \in {\cal K}_i$, we have
\begin{align}
& \mathrm{P} \left( {\cal V}^i_{k,p}(t), {\cal W}^i(t) \right) \notag \\
& \leq \mathrm{P} \left( \hat{\mu}^i_{k,p}(t) \geq \overline{\mu}^i_{k,p} + H_t, {\cal W}^i(t), {\cal B}^i(t)  \right) \notag  \\
&+ \mathrm{P} \left( \hat{\mu}^i_{k^*_i,p}(t) \leq \underline{\mu}^i_{k^*_i,p} - H_t, {\cal W}^i(t), {\cal B}^i(t) \right) 
+ \mathrm{P} ({\cal B}^i(t)^c) \notag  \\
&+ \mathrm{P} \left( \hat{\mu}^i_{k,p}(t) \geq \hat{\mu}^i_{k^*_i,p}(t), 
\hat{\mu}^i_{k,p}(t) < \overline{\mu}^i_{k,p} + H_t, \right. \notag \\
& \left. \hat{\mu}^i_{k^*_i,p}(t) > \underline{\mu}^i_{k^*_i,p} - H_t,
{\cal W}^i(t), {\cal B}^i(t)  \right)  \label{eqn:ubound1}
\end{align}
for some $H_t >0$.
We have for any suboptimal matching action $k \in {\cal L}^i_p(t)$,
\begin{align}
& \mathrm{P} \left( \hat{\mu}^i_{k,p}(t) \geq \hat{\mu}^i_{k^*_i,p}(t), 
\hat{\mu}^i_{k,p}(t) < \overline{\mu}^i_{k,p} + H_t, \right. \notag \\
& \left. \hat{\mu}^i_{k^*_i,p}(t) > \underline{\mu}^i_{k^*_i,p} - H_t,
{\cal W}^i(t), {\cal B}^i_{k,p}(t)  \right) \notag \\
&\leq \mathrm{P} \left( \hat{\mu}^{\textrm{b},i}_{k,p}(|{\cal E}^i_{k,p}(t)|) 
\geq \hat{\mu}^{\textrm{w},i}_{k^*_i,p}(|{\cal E}^i_{k^*_i,p}(t)|)
-  t^{\phi-z} , \right. \notag \\
& \left. \hat{\mu}^{\textrm{b},i}_{k,p}(|{\cal E}^i_{k,p}(t)|) < \overline{\mu}^i_{k,p} + L \left( \sqrt{d}/m_T \right)^\gamma + H_t +  t^{\phi-z}, \right. \notag \\
& \left. \hat{\mu}^{\textrm{w},i}_{k^*_i,p}(|{\cal E}^i_{k^*_i,p}(t)|) > \underline{\mu}^i_{k^*_i,p} - L \left( \sqrt{d}/m_T \right)^\gamma - H_t, \right. 
 \left. {\cal W}^i(t)    \right). \notag
\end{align}
For $k \in {\cal L}^i_p(t)$, when
\begin{align}
2 L \left( \sqrt{d}/m_T \right)^\gamma + 2H_t + 2t^{\phi-z} \leq (4 L d^{\gamma/2}+6) t^{-z/2} 
\label{eqn:boundcond}
\end{align}
the three inequalities given below
\begin{align*}
& \underline{\mu}_{k^*_i,p} - \overline{\mu}^i_{k,p} > (4 L d^{\gamma/2}+6) t^{-z/2} \\
& \hat{\mu}^{\textrm{b},i}_{k,p}(|{\cal E}^i_{k,p}(t)|) < \overline{\mu}^i_{k,p} + L \left(  \sqrt{d}/m_T \right)^\gamma + H_t + t^{\phi-z} \\
& \hat{\mu}^{\textrm{w},i}_{k^*_i,p}(|{\cal E}^i_{k,p}(t)|) > \underline{\mu}^i_{k^*_i,p} - L \left(  \sqrt{d}/m_T \right)^\gamma - H_t
\end{align*}
together imply that 
%
%\begin{align*}
$\hat{\mu}^{\textrm{b},i}_{k,p}(|{\cal E}^i_{k,p}(t)|) < \hat{\mu}^{\textrm{w},i}_{k^*_i,p}(|{\cal E}^i_{k,p}(t)|) -  t^{\phi-z}$,
%\end{align*}
%
which implies that
\begin{align}
& \mathrm{P} \left( \hat{\mu}^i_{k,p}(t) \geq \hat{\mu}^i_{k^*_i,p}(t), 
\hat{\mu}^i_{k,p}(t) < \overline{\mu}^i_{k,p} + H_t, \right. \notag \\
& \left. \hat{\mu}^i_{k^*_i,p}(t) > \underline{\mu}^i_{k^*_i,p} - H_t,
{\cal W}^i(t), {\cal B}^i_{k,p}(t)  \right) = 0. \label{eqn:vktbound1}
\end{align}
Let $H_t =  2 t^{\phi-z} + L d^{\gamma/2} m^{-\gamma}_T$. A sufficient condition that implies (\ref{eqn:boundcond}) is
\begin{align}
& 4 L d^{\gamma/2} t^{- \kappa \gamma} + 6 t^{\phi-z} \leq (4 L d^{\gamma/2}+6) t^{-z/2} \label{eqn:maincondition}
\end{align}
which holds for all $t \geq 1$ when $\phi=z/2$ and $\kappa \gamma \geq z/2$. 
Using a Chernoff-Hoeffding bound, for any $k \in {\cal L}^i_{p_i(t)}(t)$, since on the event ${\cal W}^i(t)$, $|{\cal E}^i_{k,p_i(t)}(t)| \geq t^z \log t$, we have
\begin{align}
\mathrm{P} \left( \hat{\mu}^i_{k,p}(t) \geq \overline{\mu}^i_{k,p} + H_t, {\cal W}^i(t), {\cal B}^i(t) \right) \leq t^{-2}  \label{eqn:vktbound22}
\end{align}
and
\begin{align}
\mathrm{P} \left( \hat{\mu}^i_{k^*_i,p}(t) \leq \underline{\mu}^i_{k^*_i,p} - H_t, {\cal W}^i(t), {\cal B}^i(t) \right) \leq t^{-2}. \label{eqn:vktbound32}
\end{align}

Since $\{ {\cal B}^i_{j,p}(t)^c, {\cal W}^i(t)  \} = \{ X^i_{j,p}(t) \geq t^\phi \}$, by
applying the Markov inequality, we have
\begin{align*}
\mathrm{P} ({\cal B}^i_{j,p}(t)^c, {\cal W}^i(t)) \leq \mathrm{E} [X^i_{j,p}(t)]  t^{-\phi} .
\end{align*}
Since
\begin{align*}
X^i_{j,p}(t) = \sum_{t'=1}^{|{\cal E}^i_{j,p}(t)|} \mathrm{I}(\Xi^i_{j,p}(t'))
\end{align*}
and
\begin{align*}
& \mathrm{P} \left( \Xi^i_{j,p}(t) \right) \\
&\leq \sum_{m \in {\cal C}^j_p(t) } \mathrm{P} \left( \bar{r}^j_{m,p}(t) \geq \bar{r}^{j}_{c^*_j, p}(t) \right) \\
&\leq \sum_{m \in  {\cal C}^j_p(t) }
\left(  \mathrm{P} \left( \bar{r}^j_{m,p}(t) \geq \overline{\pi}_{m,p} + H_t, {\cal W}^i(t) \right) \right. \\  
& \left. + \mathrm{P} \left( \bar{r}^{j}_{c^*_j, p}(t) \leq \underline{\pi}_{c^*_j,p} - H_t, {\cal W}^i(t) \right)  
+ \mathrm{P} \left( \bar{r}^j_{m,p}(t) \geq \bar{r}^j_{c^*_j,p}(t), \right. \right. \\
& \left. \left. \bar{r}^j_{m,p}(t) < \overline{\pi}_{m,p} + H_t,
 \bar{r}^{j}_{c^*_j,p}(t) > \underline{\pi}_{c^*_j,p} - H_t ,
{\cal W}^i(t) \right)  \right).
\end{align*}
When (\ref{eqn:maincondition}) holds, the last probability in the sum above is equal to zero while the first two probabilities are upper bounded by $e^{-2(H_t)^2 t^z \log t}$.
Thus, we have
\begin{align*}
\mathrm{P} \left( \Xi^i_{j,p}(t) \right) \leq \sum_{m \in  {\cal C}^j_p(t) } 2 e^{-2(H_t)^2 t^z \log t} \leq
2 |{\cal C}_{j}| t^{-2}.
\end{align*}
This implies that 
\begin{align*}
\mathrm{E} [X^i_{j,p}(t)] \leq \sum_{t'=1}^{\infty} \mathrm{P} (\Xi^i_{j,p}(t')) \leq 2 |{\cal C}_{j}| \sum_{t'=1}^\infty (t')^{-2} .
\end{align*}
Therefore, by the Markov inequality and union bound we get
\begin{align}
\mathrm{P} ({\cal B}^i_{j,p_i(t)}(t)^c, {\cal W}^i(t)) 
&= \mathrm{P} (X^i_{j,p_i(t)}(t) \geq t^\phi) \notag \\
&\leq 2 |{\cal C}_{j}| \beta_2 t^{-z/2} \notag 
\end{align}
and
\begin{align}
\mathrm{P} ({\cal B}^i(t)^c, {\cal W}^i(t)) 
&\leq 2 M C_{\max} \beta_2 t^{-z/2} . \label{eqn:selectionbound}
\end{align}

Then, using (\ref{eqn:vktbound1}), (\ref{eqn:vktbound22}), (\ref{eqn:vktbound32}) and (\ref{eqn:selectionbound}), we have 
\begin{align*}
\mathrm{P} \left( {\cal V}^i_{k}(t), {\cal W}^i(t)  \right) \leq 2 t^{-2} 
+ 2 M C_{\max}  \beta_2 t^{-z/2},
\end{align*}
for any $k \in {\cal L}^i_{p_i(t)}(t)$, and
By (\ref{eqn:subregret}), and by the result of Appendix A, we get the stated bound for $\mathrm{E} [R^s_i(T)]$.
%
%\begin{align}
%E[R_s(T)] &\leq 2^d T^{\kappa d}
% \left( 2 (M-1+|{\cal F}_i|) \beta_2  \right. \notag \\
%& \left. +  2 (M-1) C_{\max} \beta_2 \sum_{t=1}^T \frac{1}{t^{1-z/2}} %\right) \notag \\
%&\leq  2^{d+1} (M-1+|{\cal F}_i|) \beta_2 T^{\kappa d} \notag \\
%&+ \frac{2^{d+2} (M-1)  C_{\max} \beta_2}{z} T^{\kappa d + z/2}, %\label{eqn:regret_s}
%\end{align}
%
%where (\ref{eqn:regret_s}) follows from . 
%
\end{proof}

The next lemma bounds $\mathrm{E} [R^n_i(T)]$.

\begin{lemma} \label{lemma:nearoptimal}
Consider all CAs running DISCOM with parameters $H_1(t) = t^{z} \log t$, $H_2(t) = C_{\max} t^{z} \log t$, $H_3(t) = t^{z} \log t$ and $m_T = \left\lceil T^{\kappa} \right\rceil$, where $0<z<1$ and $\kappa = z/ (2\gamma)$. Then, we have
\begin{align*}
\mathrm{E} [R^n_i(T)] \leq \frac{ (14 L d^{\gamma/2} + 12) }{1-z/2} T^{1-z/2} + 4 C_{\max} \beta_2.
\end{align*}
\end{lemma}
\begin{proof}
At any time $t$, for any $k \in  {\cal K}_i -  {\cal L}^i_p(t)$ and $x \in p$, we have 
\begin{align}
\mu^i_{ k^*_i(x) }(x) - \mu^i_k(x) \leq (7 L d^{\gamma/2} + 6) t^{-z/2} .      \notag
\end{align}
Similarly, for any $j \in {\cal M}$, $c \in {\cal C}_j -  {\cal C}^j_p(t)$ and $x \in p$, we have 
\begin{align}
\pi_{ c^*_j(x) }(x) - \pi_c(x) \leq (7 L d^{\gamma/2} + 6) t^{-z/2}.      \notag
\end{align}

Due to the above inequalities, if a near optimal action in ${\cal C}_i \cap  ({\cal K}_i -  {\cal L}^i_p(t))$ is chosen by CA $i$ at time $t$, the contribution to the regret is at most 
$(7 L d^{\gamma/2} + 6) t^{-z/2}$.
If a near optimal CA $j \in {\cal M}_{-i} \cap  ({\cal K}_i -  {\cal L}^i_p(t) )$ is called by CA $i$ at time $t$, and if CA $j$ selects one of its near optimal contents in ${\cal C}_j - {\cal C}^j_p(t)$, then the contribution to the regret is at most $2 (7 L d^{\gamma/2} + 6) t^{-z/2}$.
Moreover since we are in an exploitation step, the near-optimal CA $j$ that is chosen can choose one of its suboptimal contents in ${\cal C}^j_p(t)$ with probability at most $2 C_{\max} t^{-2}$, which will result in an expected regret of at most $4 C_{\max} t^{-2}$. 

Therefore, the total regret due to near optimal choices of CA $i$ by time $T$ is upper bounded by 
\begin{align*}
& (14 L d^{\gamma/2} + 12) \sum_{t=1}^T t^{-z/2} + 4 C_{\max} \sum_{t=1}^T t^{-2} \\
& \leq \frac{ (14 L d^{\gamma/2} + 12) }{1-z/2} T^{1-z/2} + 4 C_{\max} \beta_2 .
\end{align*}
by using the result in Appendix \ref{app:seriesbound}. 
\end{proof}

Next, we give the proof ot Theorem 1 by combining the results of the above lemmas. 

%\rev{From Lemma \ref{lemma:nearoptimal}, we see that the regret due to near optimal matching actions depends exponentially on $\theta$ which is related to negative of $\kappa$ and $z$. Therefore $\kappa$ and $z$ should be chosen as large as possible to minimize the regret due to near optimal arms.} 

\vspace{-0.1in}
\subsection{Proof of Theorem 1} 
The highest orders of regret that come from Lemmas 1, 2 and 3 are $\tilde{O}(T^{\kappa d + z})$, $O(T^{1- z/2})$, $O(T^{1-z/2})$. We need to optimize them with respect to the constraint
Regret is minimized when $\kappa d + z = 1 - z/2$, which is attained by $z = 2\gamma/(3\gamma+d)$. 
Result follows from summing the bounds in Lemmas \ref{lemma:explorations}, \ref{lemma:suboptimal1} and \ref{lemma:nearoptimal}.

 \section{Proof of Corollary \ref{cor:confidence}} \label{app:confidence}
 
From the proof of Lemma 2, for $z=2\gamma/(3\gamma+d)$, we have 
\begin{align}
\mathrm{P} \left( {\cal V}^i_{k}(t), {\cal W}^i(t)  \right) 
\leq 2 t^{-2} 
+ 2 M C_{\max} \beta_2  t^{-\gamma/(3\gamma+d)}       \notag
\end{align}
for $k \in {\cal L}^i_{p_i(t)}(t)$.
This implies that 
\begin{align}
& \mathrm{P} \left( a_i(t) \in  {\cal L}^i_{p_i(t)}(t) , {\cal W}^i(t)  \right) \notag \\
& \leq \hspace{-0.2in} \sum_{k \in {\cal L}^i_{p_i(t)}(t) } \mathrm{P} \left( {\cal V}^i_{k}(t), {\cal W}^i(t)  \right)  \notag \\
& \leq \frac{2 |{\cal K}_i|}{t^2} + \frac{2 |{\cal K}_i| M C_{\max} \beta_2}{t^{\gamma/(3\gamma+d)}}. \notag
\end{align} 

The difference between the expected reward of an action within a hypercube from its expected reward at the center of the hypercube is at most $L d^{\gamma/2}/(m_T)^\gamma$. Since $m_T = \lceil T^{1/(3\gamma+d)} \rceil$, $a_i(t) \in {\cal K}_i - {\cal L}^i_{p_i(t)}(t)$ implies that 
\begin{align}
\mu^i_{a_i(t)}(x_i(t)) & \geq \mu^i_{k^*_i(x_i(t))}(x_i(t))  - 
(6Ld^{\gamma/2} + 6) T^{-\gamma/(3\gamma+d)}   .   \notag
\end{align}

 \section{Proof of Theorem \ref{thm:nolabel}} \label{app:theorem2}

In order for time $t$ to be an exploitation slot for CA $i$ it is required that ${\cal M}^{\textrm{ut}}_{i,p_i(t)}(t) \cup {\cal M}^{\textrm{ue}}_{i,p_i(t)}(t) \cup {\cal C}^{\textrm{ue}}_{i,p_i(t)}(t)  = \emptyset$.
Since the counters of DISCOM are updated only when feedback is received, and since the control functions are the same as the ones that are used in the setting where feedback is always available, the regret due to suboptimal and near optimal matching actions by time $t$ with missing feedback will not be any greater than the regret due to suboptimal and near optimal matching actions for the case when the users always provide feedback. 
Therefore, the bounds given in Lemmas 2 and 3 will also hold for the case with missing feedback. Only the regret due to trainings and explorations increases, since more trainings and explorations are needed before the counters exceed the values of the control functions such that the relevance score estimates are accurate enough to exploit. 
Consider any $p \in {\cal P}_T$. From the definition of DISCOM, the number of exploration slots in which content $c \in {\cal C}_i$ is matched with CA $i$'s user and the user's feedback is observed is at most $\lceil T^{2\gamma/(3\gamma+d)}\rceil$. The number of training slots in which CA $i$ requested content from CA $j \in {\cal M}_{-i}$ and received the feedback about this content from its user is at most $\left\lceil C_{\max} T^{2\gamma/(3\gamma+d)} \log T \right\rceil$. The number of exploration slots in which CA $i$ selected CA $j \in {\cal M}_{-i}$ is at most $\left\lceil T^{2\gamma/(3\gamma+d)} \log T \right\rceil$. 

Let $\tau_{\textrm{exp}}(T)$ be the random variable which denotes the smallest time step for which for each $c \in {\cal C}_i$ there are $\lceil T^{2\gamma/(3\gamma+d)}\rceil$ feedback observations, for each $j \in {\cal M}_{-i}$ there are $\left\lceil C_{\max} T^{2\gamma/(3\gamma+d)} \log T \right\rceil$ feedback observations for the trainings and $\left\lceil T^{2\gamma/(3\gamma+d)} \log T \right\rceil$ feedback observations for the explorations. 
Then, $\mathrm{E} [\tau_{\textrm{exp}}(T)]$ is the expected number of training plus exploration slots by time $T$. Let $Y_{\textrm{exp}}(t)$ be the random variable which denotes the number of time slots in which the feedback is not provided by the users of CA $i$ till CA $i$ received $t$ feedbacks from its users. Let $A_i(T) = Z_i T^{2\gamma/(3\gamma+d)} \log T +  (|{\cal C}_i| + 2(M-1))$. We have
\begin{align*}
\mathrm{E}[\tau_{\textrm{exp}}(T)] = \mathrm{E} [Y_{\textrm{exp}}(A_i(T))] + A_i(T).
\end{align*}
$Y_{\textrm{exp}}(A_i(T))$ is a negative binomial random variable with probability of observing no feedback at any time $t$ equals to $1-p_r$. Therefore,
\begin{align}
\mathrm{E} [Y_{\textrm{exp}}(A_i(T))] = (1-p_r)A_i(T)/p_r.      \notag
\end{align}
Using this, we get
\begin{align*}
\mathrm{E} [\tau_{\textrm{exp}}(T)] = A_i(T)/p_r.
\end{align*}
The regret bound follows from substituting this into the proof of Theorem 1.

 \section{Proof of Theorem \ref{thm:adaptivemain2}} \label{app:theorem3}
 The basic idea is to choose $\tau_h$ in a way that the regret due to variation of relevance scores over time and the regret due to variation of estimated relevance scores due to the limited number of observations during each round is balanced. 
Majority of the steps of this proof is similar to the proof of Theorem 1 hence some of the steps are omitted. 

Consider a round $\eta$ of length $2\tau_h$. Denote the set of time slots in round $\eta$ by $[\eta]$. 
For any $c \in {\cal C}$ let
\begin{align}
\bar{\pi}_{c,p,\eta} &:= \sup_{x \in p, t \in [\eta]} \pi_{c,t}(x) ,  \notag \\
\underline{\pi}_{c,p,\eta} &:= \inf_{x \in p, t \in [\eta]} \pi_{c,t}(x) . \notag 
\end{align}
For any $k \in {\cal K}_i$ let
\begin{align}
 \bar{\mu}^i_{k,p,\eta} &:= \sup_{x \in p, t \in [\eta]} \mu^i_{k,t}(x) ,  \notag \\
 \underline{\mu}^i_{k,p,\eta} &:= \inf_{x \in p, t \in [\eta]} \mu^i_{k,t}(x) ,  \notag 
\end{align}
where $\mu^i_{k,t}(x)$ is defined as the time-varying version of $\mu^i_{k}(x)$ given in (1) under Assumption 2.
For CA $i$, the set of suboptimal matching actions is given as 
\begin{align}
{\cal L}^i_{p,\eta}(t) &:= \left\{ k \in {\cal K}_i :  \underline{\mu}^i_{k^*_i(p),p,\eta} - \bar{\mu}^i_{k,p,\eta} \right. \notag \\
&\left. > (4 L d^{\gamma/2} +6 ) ( t\textrm{ mod } \tau_h +1)^{-z/2} + \frac{4\tau_h}{T_s}  \right\} ,     \notag
\end{align}
where $k^*_i(p)$ is the matching action with the highest net reward for the context at the center of $p$ at the time slot in the middle of round $\eta$. 

Consider the regret due to explorations and trainings for $\textrm{DISCOM}_{\eta}$ incurred over times when it is in the active sub-phase (over $\tau_h$ time slots). 
Similar to the proof of Lemma 1 it can be shown that the regret due to trainings and explorations is 
\begin{align}
\mathrm{E} [R^e_i(\tau_h)] = \tilde{O} \left( \tau_h^{z+\kappa d}     \right)    .  \notag
\end{align}
Similar to the proof of Lemma 2, it can be shown that the regret due to suboptimal matching action selections is 
\begin{align}
\mathrm{E} [R^s_i(\tau_h)]  =   O \left( \tau_h^{1-z/2} \right)   \notag
\end{align}
when $\kappa = z/(2\gamma)$.
Since the definition of a sub-optimal matching action is different for dynamic user and content characteristics, the regret due to near optimal matching actions in ${\cal K}_i - {\cal L}^i_{p,\eta}(t)$ is different from Lemma 3. 
At time $t$ which is in round $\eta$, since a near optimal matching action's contribution to the one-step regret is at most 
\begin{align}
(8 L d^{\gamma/2} +12 ) (t \mod \tau_h +1)^{-z/2} + 4 \tau_h/T_s      \notag
\end{align}
summing over all time slots in a round $\eta$, we have
\begin{align}
\mathrm{E} [R^n_i(\tau_h)]  =   O \left( \tau_h^{1-z/2} \right) + O \left( \frac{\tau_h^2}{T_s} \right) .         \notag
\end{align}
Clearly we have $\mathrm{E} [R^s_i(\tau_h)] \leq \mathrm{E} [R^e_i(\tau_h)] $. 
Let $\tau_h = \lfloor T_s^\phi \rfloor$ for some $\phi>0$. Then we have
\begin{align}
\frac{ \mathrm{E} [R^e_i(\tau_h)]}{\tau_h} =  \tilde{O} \left( T_s^{\phi z+\phi \kappa d - \phi}     \right)  ,     \notag
\end{align}
and
\begin{align}
\frac{ \mathrm{E} [R^n_i(\tau_h)]}{\tau_h}  =   O \left( T_s^{-\phi z/2} \right) + O \left( T_s^{\phi-1} \right) .     \notag
\end{align}
The sum $(\mathrm{E} [R^e_i(\tau_h)] + \mathrm{E} [R^s_i(\tau_h)]  + \mathrm{E} [R^n_i(\tau_h)])/\tau_h$ is minimized by setting $z = 2\gamma/(3\gamma+d)$ and $\phi = 1/(1+z/2)$. 
Hence, $\tau_h = \lfloor T_s^{\frac{3\gamma+d}{4\gamma+d}} \rfloor$ the order of the time averaged regret is equal to 
$\tilde{O}\left( T_s^{\frac{-\gamma}{4\gamma+d}}\right)$. 
      
}

%\vspace{-0.2in}
\bibliographystyle{IEEE}
\bibliography{OSA}

\begin{IEEEbiography}
[{\includegraphics[width=1in, height=1.25in,clip,keepaspectratio]{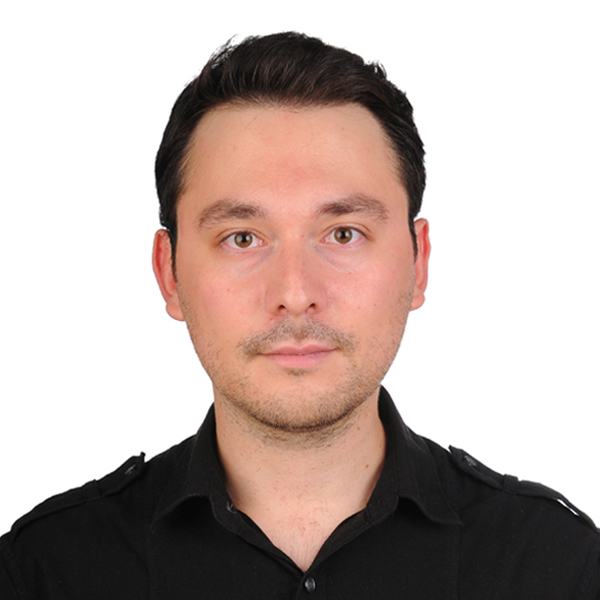}}]{Cem Tekin}
Cem Tekin is an Assistant Professor in Electrical and Electronics Engineering Department at Bilkent University, Turkey. From February 2013 to January 2015, he was a Postdoctoral Scholar at University of California, Los Angeles. He received the B.Sc. degree in electrical and electronics engineering from the Middle East Technical University, Ankara, Turkey, in 2008, the M.S.E. degree in electrical engineering: systems, M.S. degree in mathematics, Ph.D. degree in electrical engineering: systems from the University of Michigan, Ann Arbor, in 2010, 2011 and 2013, respectively. His research interests include machine learning, multi-armed bandit problems, data mining, multi-agent systems and game theory. He received the University of Michigan Electrical Engineering Departmental Fellowship in 2008, and the Fred W. Ellersick award for the best paper in MILCOM 2009.
\end{IEEEbiography}

\begin{IEEEbiography}
[{\includegraphics[width=1in, height=1.25in,clip,keepaspectratio]{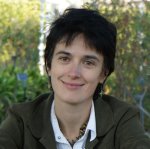}}]{Mihaela van der Schaar}
Mihaela van der Schaar is Chancellor Professor of Electrical Engineering at University of California, Los
Angeles. Her research interests include network economics and game theory, online learning, dynamic multi-user networking and
communication, multimedia processing and systems, real-time stream mining. She is an IEEE Fellow, a Distinguished Lecturer of
the Communications Society for 2011-2012, the Editor in Chief of IEEE Transactions on Multimedia and a member of the Editorial
Board of the IEEE Journal on Selected Topics in Signal Processing. She received an NSF CAREER Award (2004), the Best Paper
Award from IEEE Transactions on Circuits and Systems for Video Technology (2005), the Okawa Foundation Award (2006), the
IBM Faculty Award (2005, 2007, 2008), the Most Cited Paper Award from EURASIP: Image Communications Journal (2006), the
Gamenets Conference Best Paper Award (2011) and the 2011 IEEE Circuits and Systems Society Darlington Award Best Paper Award.
She received three ISO awards for her contributions to the MPEG video compression and streaming international standardization
activities, and holds 33 granted US patents.
\end{IEEEbiography}

\end{document}